\documentclass[3p]{elsarticle}
\usepackage{amsmath}
\usepackage{amsfonts}
\usepackage{amssymb}
\usepackage{amsthm}
\usepackage[parent]{currfile}
\usepackage{xcolor}

\usepackage{tikz}
\usetikzlibrary{decorations.pathreplacing,positioning, arrows.meta}
\usepackage{multirow}
\usepackage{multicol}
\usepackage[utf8]{inputenc}
\usepackage{float}
\usepackage{graphicx}
\usepackage{pdfpages}

\usepackage{lineno,hyperref}
\usepackage{subfig}
\modulolinenumbers[5]

\newtheorem{lemma}{Lemma}
\newtheorem{proposition}{Proposition}
\usepackage[mode=buildnew]{standalone}
\usepackage{tikz}
\usepackage{appendix}
\usepackage{listings}
\usepackage{slashbox}
\usepackage{caption}
\usepackage{colortbl}
\definecolor{gris0}{gray}{0.98}
\definecolor{gris1}{gray}{0.93}
\definecolor{gris2}{gray}{0.84}
\definecolor{gris3}{gray}{0.75}
\definecolor{gris4}{gray}{0.69}
\definecolor{gris5}{gray}{0.63}
\definecolor{gris6}{gray}{0.57}

\begin{document}

\begin{frontmatter}

\title{{\color{black} Maintenance cost assessment for heterogeneous multi-component systems \\ incorporating perfect inspections and waiting time to maintenance}}


\author{Lucía Bautista} 
\address{Department of Mathematics \\ University of Extremadura, Cáceres (Spain) \\ email: luciabb@unex.es}

\author{Inma T. Castro \corref{mycorrespondingauthor}}
\address{Department of Mathematics \\ University of Extremadura, Cáceres (Spain) \\ email: inmatorres@unex.es}

\author{Luis Landesa} 
\address{Department of Computers and Communications Technology \\ University of Extremadura, Cáceres (Spain) \\ email: llandesa@unex.es}




\cortext[mycorrespondingauthor]{Corresponding author}


\begin{abstract}

 
{\color{black} Most existing research about complex systems maintenance assumes they consist of the same type of components. However, systems can be assembled with heterogeneous components (for example degrading and non-degrading components) that require different maintenance actions. Since industrial systems become more and more complex, more research about the maintenance of systems with heterogeneous components is needed. {\color{black} For this reason, in this paper, a system consisting of two groups of components: degrading and non-degrading components is analyzed. The main novelty of this paper is the evaluation of a maintenance policy at system-level coordinating condition-based maintenance for the degrading components, delay time to the maintenance and an inspection strategy for this heterogeneous system. } {\color{black} To that end, an analytic cost model is built using the semi-regenerative processes theory}. Furthermore, a safety constraint related to the reliability of the degrading components is imposed. To find the optimal maintenance strategy, meta-heuristic algorithms are used.}

\end{abstract}

\end{frontmatter}


\section{Introduction}
Maintenance of multi-component systems is a key challenge usually conditioned by dependencies between components. In the framework of multi-component systems, in this paper, we focus on three important aspects of these systems that affect to their maintenance and increase the variability in maintenance modeling. These aspects are the heterogeneity of the components, the opportunistic maintenance and the lead time to start a maintenance action. \\

{\bf 1. Heterogeneous components}. Multi-components systems can be assembled 
of heterogeneous components that require different maintenance strategies \cite{Minou}. It is the case when a system is composed of mechanical and electronic components. Electronic component failures occur without warning. It contrasts with perceptible deterioration that are frequent precursors of failure in mechanical components \cite{Buchacker}. Examples of these multi-component systems with electronic and mechanical parts can be found in a cooling system and in a lithography system. 
\begin{itemize}
\item Cooling systems play a key role in many industrial applications. Because of their electromechanical structure, components of these cooling systems are split into degrading and non-degrading parts. Mechanical parts are related to bearings, rotor or blades. The fan's electronic parts are considered as non-degrading \cite{Heier2} \cite{Buchacker} and these components are grouped in the so-called electronic control unit (ECU). 
\item A lithography system consists of components subject to a continuous deterioration (certain belts that transport paper inside the printer or filters that 
slowly get clogged). Furthermore, lithography systems also consist of electronic components that do not give any warning before they fail \cite{Zhu}.        
\end{itemize}
Maintenance policies with heterogeneous components are validated and applied at an OEM (Original Equipment Manufacturer) in the compressed air, generator and pump industry \cite{Poppe}. \\

{\bf 2.  Opportunistic maintenance}. In the maintenance of multi-component systems, opportunistic maintenance plays an important role in reducing the maintenance cost and avoiding unnecessary shut-downs. 
This maintenance stems from the fact that the cost of simultaneous
maintenance actions on various components could
be less than the total cost of individual maintenance
actions. Early works on opportunistic maintenance applied this maintenance technique to different systems such as rocket engines, manned air-crafts and ballistic missile systems \cite{Samat}. {\color{black} Interactions between opportunistic maintenance and condition-based maintenance have been studied by different authors (\cite{Koochaki})}. \\

{\bf 3. Lead time}. In many practical situations, there is a lead time between the failure of a component and the time of its maintenance. This lead time is due to different factors. For example, repairmen are not continuously available or spare parts may not be on stock and have to be ordered \cite{Jonge2}. Literature that deals with delayed time (or lead time) in maintenance operations for degrading systems can be found in \cite{MeierHirmer} and \cite{Zhao} among others. \\

In this paper, a multi-component system  with $m$ degrading components and $n$ non-degrading components is {\color{black} studied} ($m, n >0$). The degrading components are subject to a continuous deterioration modeled {\color{black} as a gamma process} \cite{CastroPr}. Non-degrading component failures occur without evidence of deterioration. {\color{black} We assume in this paper that } all non-degrading components are treated as one part (the so-called single electronic control unit (ECU) in some systems\cite{Heier2}). {\color{black} For the sake of simplicity, we call this part the ``non-degrading part''}. An exponential distribution models the time between failures of the non-degrading {\color{black} part} (\cite{Zhu} and \cite{Poppe}). The use of the exponential distribution as time between failures of the non-degrading {\color{black} part} is justified by the Palm-Khintchine theorem. This theorem states that {\color{black} the superposition of a large number of independent equilibrium renewal processes, each with small intensity, is asymptotically a Poisson process (\cite{Heyman}, pages 156-157). This theorem provides us a theoretical justification for using the exponential distribution for the failure times of the {\color{black} non-degrading part of the system} assuming a sufficient number of independent non-degrading components $n$ with intensities fulfilling the assumptions of the Palm-Khintchine theorem. These assumptions are that, as $n$ increases, the asymptotic rate of the combined process is constant and the combined process has renewals very infrequently. To obtain a theoretical development of this theorem, the reader is referred to \cite{Heyman}}  
{Poppe et. al \color{black} \cite{Poppe} validated how many non-degrading components are required to assume that the time between failures of the {\color{black} non-degrading part} is exponentially distributed.}

When a {\color{black} degrading component (or the non-degrading part) fails}, a signal is sent to the maintenance team. The maintenance team takes $\tau$ units of time to arrive on site and next perform a corrective maintenance. 
This lead time can be justified as the time to arrive on site and/or the time to deliver the failed components to its supplier. {After \color{black} a corrective maintenance, the degrading component (or the non-degrading part) can consider as a new one.    }

In an inspection time, there is not lead time to start the maintenance. The justification is that, since the inspection is scheduled beforehand, maintenance team has sufficient time to overcome the need of new components or the time to arrive on situ. These maintenance times give the maintenance team an opportunity to check the state of the rest of the components. The degradation processes of the degrading components are monitored and this information is used to decide if the component should be correctively or preventively replaced by introducing thresholds on the degradation level. In the case of non-degrading {\color{black} part}, {\color{black} while another component is being maintained}, maintenance team checks {\color{black} its state}. If {\color{black} it is failed, a maintenance corrective is performed}.

To reduce the downtime of the degrading components, inspections are performed. These inspection times are also opportunities of maintenance and some degrading components can be preventively replaced if they are too degraded or correctively replaced if they are failed or left as they are.

Each maintenance action implies a cost and a reward function is also included in this model. 
The objective is to analyze the time between inspections and the preventive thresholds of the degrading components that minimize a given objective cost function. The complete replacement of the system (renewals) can be seen as a regenerative process and renewal techniques can be applied to analyze the optimal maintenance strategy. But, describing the functioning of the system using renewal theory can be tedious. {\color{black} To deal with it, the asymptotic behaviour of the process that models the maintained system is reduced using semi-regenerative techniques}. The use of these techniques simplifies the analysis and reduces the computation time \cite{Castro}. Conditions are given in this paper to guarantee the application of the semi-regeneration theory. {\color{black} Our paper is not the first to introduce the use of semi-regenerative techniques in the {\color{black}multi-component system} analysis: two-unit series system \cite{Castanier} and multi-unit systems with identical components \cite{Zhang} are examples of the use of these techniques in maintenance literature.  }

The complexity of optimizing the maintenance strategy increases when the number of components increases. Meta-heuristic algorithms, such as genetic algorithms (GA) \cite{Compare}, colony algorithms \cite{Samrout} and simulated annealing algorithms \cite{Rasmekomen}, are widely employed to optimize maintenance models. Furthermore, to control the risk of a failure of all the degrading components, a requirement is specified. This requirement imposes that the probability that all the degrading components fail between two successive maintenance actions should not exceed a fixed probability limit. The optimization of this maintenance strategy is performed under this requirement.

The model that we present in this article is inspired by \cite{Castro}. However, there are three important differences. 

\begin{itemize}
\item Whereas \cite{Castro} considered a continuous monitoring of the system, in this paper an inspection policy is scheduled. 
\item In \cite{Castro}, when a component fails it is immediately replaced. In this paper, when {\color{black} a failure happens between inspections}, a waiting time to the replacement is imposed. Model shown in \cite{Castro} is a particular case of the model shown in this paper when the time between inspections tends to infinity and delay time is equal to zero. 
\item In the search of the optimal maintenance strategy, meta-heuristic algorithms are used in this paper whereas in \cite{Castro} the optimal maintenance policy is obtained by visualization of the objective cost function. 
\end{itemize}

The paper is organized as follows. In Section \ref{description} the behavior of the multi-component  system is modeled, describing its different components. The theoretical probabilistic background is developed in detail in Sections \ref{theory} and \ref{probability}. Numerical examples are given in Section \ref{numerical}. Section \ref{conclusions} concludes.

\section{System description} \label{description}
This section describes the functioning of the system and the main assumptions that the model fulfills. 

\subsection{General assumptions}
A multi-component system consisting of $m$ degrading components and $n$ non-degrading components is analyzed. {\color{black} We assume in this paper that the $n$ non-degrading components are treated as one part, the so-called non-degrading part}. 
\begin{enumerate}
    \item The $m$ degrading components are subject to a continuous deterioration following a gamma process.  Let $X_i(t)$ be the deterioration of component $i$ at time $t$ for $i=1, 2, \ldots, m$. For $s<t$, the density function of $X_i(t)-X_i(s)$ is given by:
    \begin{equation}\label{fgamma}
    f_{\alpha_i (t-s),\beta_i}(x)= \frac{{\beta_i}^{\alpha_i (t-s)}}{\Gamma({\alpha_i (t-s)})} x^{\alpha_i (t-s)-1} e^{-\beta_ix}, \hspace{1cm} x \geq 0,
    \end{equation}
     for $i \in I^m$, where $I^{m}=\left\{1, 2, \ldots, m\right\}$ and where $\Gamma(\cdot)$ is the well-known gamma function. {\color{black} These degradation processes grow independently each other.  }
\item When the deterioration level of a degrading component exceeds a failure threshold, this component fails and a signal is sent to the maintenance team. Maintenance team takes $\tau$ time units to start the corrective replacement. {\color{black} The system continues working during the delay time $\tau$ until the maintenance team performs the corresponding replacement}. Component breakdowns can occur between the signal time and the onset of the maintenance action. The corrective maintenance time is used as an opportunity to check the {\color{black} state of the rest of components}. At the corrective maintenance time, if the deterioration level of the degrading component $i$, with $i \in I^m$, exceeds the preventive threshold $M_i$ but it is less than the failure threshold $L_i$, a preventive replacement of this component is performed. At the corrective maintenance time, if the deterioration level of the degrading component $i$ exceeds the failure threshold $L_i$, a corrective replacement of this component is performed. {\color{black} If a failure of the non-degrading {\color{black} part} occurs in the waiting time to maintenance, a corrective maintenance of this component is performed}.
    \item Let $Y$ be the time between failures of the non-degrading {\color{black} part}. It follows an exponential distribution with parameter $\lambda$ and survival function
    $\bar{F}_Y(t)= \exp(-\lambda t)$. 
    \item When the non-degrading {\color{black} part fails}, a signal is sent to the maintenance team and it takes $\tau$ time units to start the corrective maintenance. {\color{black} As before, the system continues working during the delay time $\tau$ until the maintenance team performs the corresponding maintenance tasks}. At the time of the maintenance of the non-degrading {\color{black} part}, if the deterioration level of the degrading component $i$, $i \in I^m$, exceeds the preventive threshold $M_i$ but it is less than $L_i$, a preventive maintenance is performed. If the deterioration level of the degrading component $i$ exceeds the failure threshold $L_i$, a corrective maintenance is performed. Corrective and preventive maintenance imply the replacement of the component by a new one.

    \item Inspections are performed to check the state of the system. In these inspection times, if the deterioration level of the degrading component $i$, $i \in I^m$, exceeds the preventive threshold $M_i$ but it is less than $L_i$, a preventive maintenance on the component $i$ is produced. If {\color{black}a degrading component or the non-degrading part} is failed in an inspection time, a corrective maintenance is performed.   
    
\item A sequence of costs and rewards is imposed in this model. When degrading components are working, they provide a reward that decreases with the deterioration level of the component. A classical exponential function given by \cite{Niese} is used in this paper. Let $g_i$ be the reward function for the degrading component $i$. Given the deterioration level $x$ of the degrading component $i$, the reward function $g_i$ is given by
\begin{equation*}
g_i(x)=\theta_0+g \exp{(-\gamma_i x)}, \quad 0 \leq x \leq L_i, \quad \gamma_i >0, \quad i \in I^m, 
\end{equation*}
where $\theta_0, g \geq 0$. If $\gamma_i=0$, the reward is constant.

A corrective (preventive) replacement of the component $i$ implies a cost of $C_i^c$ ($C_i^p$) monetary units (m.u.). The corrective maintenance of the non-degrading {\color{black} part} implies a fixed cost of $C^f$ m.u. When the degrading component $i$ is down, $i \in I^m$, a cost of $c_i$ monetary units (m.u.) per {\color{black} time units} (t.u.) is implied. When the non-degrading {\color{black} part is} down, a cost of $c^{nm}$ m.u. per t.u. is produced. Each inspection costs $C^I$ m.u. 
\end{enumerate}
{\color{black} By maintenance time we mean the time in which maintenance actions are performed}. The different maintenance actions in a maintenance time are the following. 
    \begin{itemize}
        \item  \textbf{Degrading components}
           \begin{enumerate}
            \item If the deterioration level of a degrading component exceeds the failure threshold, this component is correctively replaced. 
            \item If the deterioration level of a degrading component exceeds its preventive threshold but it is less than the failure threshold, a preventive replacement is performed. 
            \item The degrading component is left as it is if its deterioration level is lower than its preventive threshold.
         \end{enumerate}
        \item \textbf{Non-degrading part}
        \begin{enumerate}
            \item If the {\color{black} non-degrading part is failed, it is correctively maintained}.
        \end{enumerate}
    \end{itemize}

{\subsection{Inspection policy}}
{\color{black} An inspection policy is integrated in this model}. We assume that the inspections remain fixed over time and they are performed at periodic times $T, 2T, 3T, \ldots$ with time between inspections $T_{k-1}-T_k=T>0$.   

In absence of maintenance, a new system starts working at time $t=0$. Let $Z$ be the time to the first failure, if   
$$T_{k-1} < Z < Z+\tau<T_k,$$ 
then the maintenance action is performed at time $Z+\tau$. However, if {\color{black} the failure arrives at time $Z$} with 
$$T_{k-1} < Z <T_k <Z+\tau, $$
the maintenance action is performed at time $T_k$ since an inspection time was beforehand scheduled at time $T_k$. Hence, the maintenance action is performed at time $\min(Z+\tau,T_k)$.  

Figure \ref{tikz2} represents an scheme of the inspection policy. Inspection times are denoted by $T_1, T_2, \ldots$ and maintenance times by $O_1, O_2, \ldots$. Orange arrows represent time between maintenance actions. 


\begin{figure}[htbp]
\centering
\includegraphics[scale=0.85]{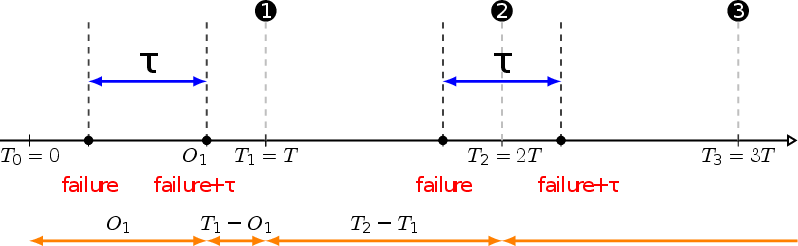}
\caption{Inspection policy.} \label{tikz2}
\end{figure}

{\color{black} Figure \ref{paths2} shows a realization of the evolution of the maintained system with two identical degrading components (preventive threshold $M$, failure threshold $L$). Let $T_1, T_2, \ldots $ be the inspection times and let $O_1, O_2, \ldots$ be the maintenance times. In the first inspection $T_1=T$, a corrective replacement of component 1 and a preventive replacement of component 2 are performed. In the second inspection time, both components are left as they are. A maintenance time occurs in $(T_2,T_3)$ due to the failure of the non-degrading {\color{black} part}. Finally, in $T_4$, the corrective replacements of the two degrading units are performed. }

 \begin{figure}[htbp]
    \includegraphics[scale=0.4]{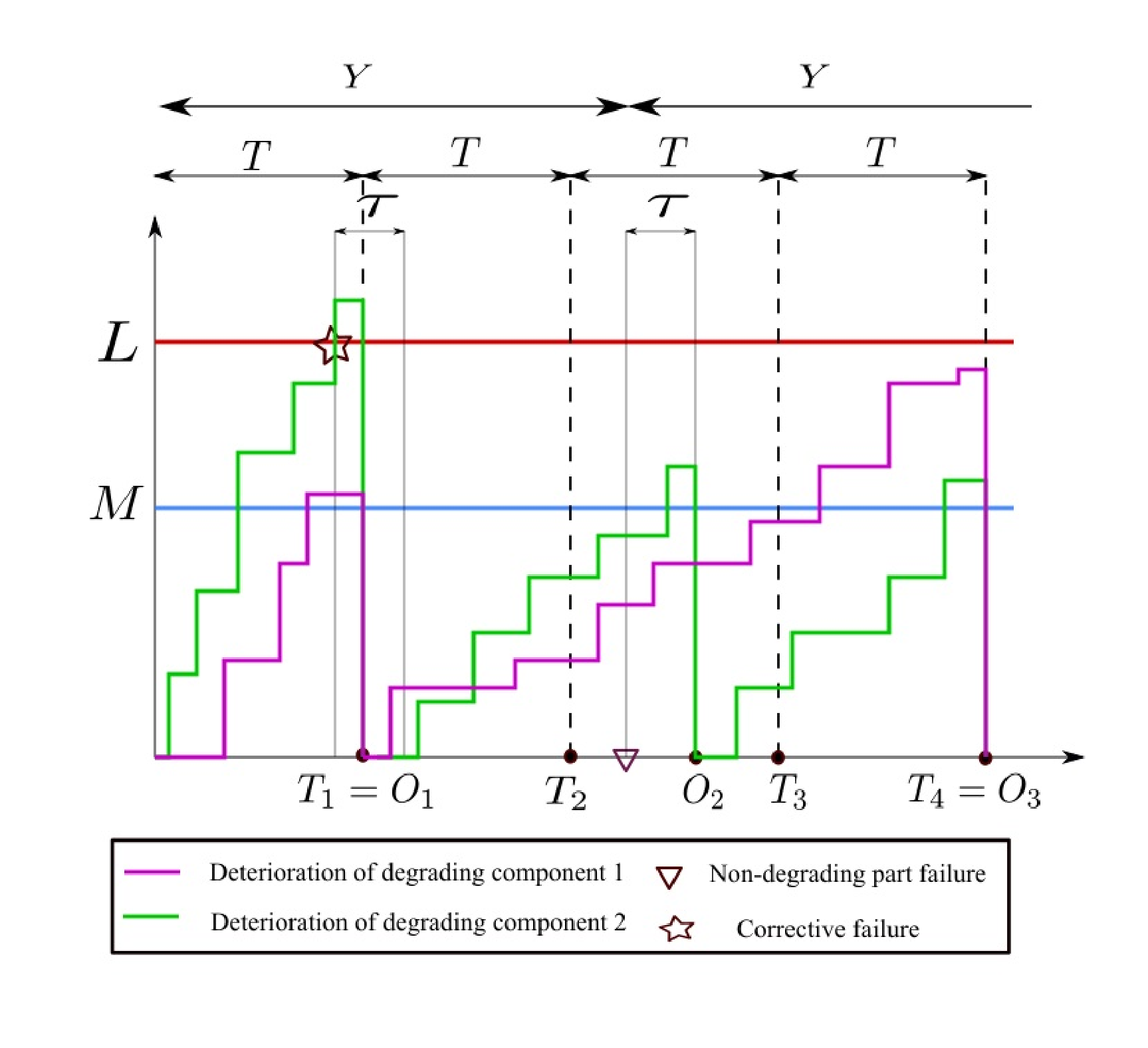}
   \caption{Paths of the maintained system.} \label{paths2}
   \end{figure}

\section{Theoretical development} \label{theory}

Let ${\bf W}(t)$ be the degradation of the degrading components at time $t$:
\begin{equation*} 
 {\bf W}(t)=\left(X_1(t), X_2(t), \ldots, X_m(t)\right), \quad t \geq 0, 
\end{equation*}
where $X_i(t)$ denotes the degradation of degrading component $i$ at time $t$ with $i \in I^m$.

\subsection{Time to a renewal} 
Notation $\sigma_{z_i}$ stands for the 
first time at which the gamma process $\left\{X_i(t), t \geq 0\right\}$ exceeds a degradation level $z_i$. We recall that, for a fixed $z_i$ and for a gamma process with parameters $\alpha_i$ and $\beta_i$, the distribution of $\sigma_{z_{i}}$ is given by
\begin{align*}
 F_{\sigma_{z_i}}(t) = \int_{z_i}^{\infty} f_{\alpha_i t, \beta_i}(x)dx = \frac{\Gamma(\alpha_i t, z_i\beta_i)}{\Gamma(\alpha_i t)} \hspace{1cm}  t \geq 0, 
\end{align*}
where $f_{\alpha_i t, \beta_i}$ denotes the density function of a gamma distribution with parameter $\alpha_i t$ and $\beta_i$ given in Eq. (\ref{fgamma}) and $\Gamma(\alpha_i t, z_i \beta_i)$ denotes the  incomplete gamma function. For subsequent analysis, the distribution of $\sigma_{L_i}-\sigma_{M_i}$ is used in this paper. According to \cite{Castro2}, the survival function of this variable is given by  
\begin{align*}   \label{doblefunction}
\bar{F}_{\sigma_{L_i}-\sigma_{M_i}} (t) = \int_{x=0}^{\infty}\int_{y=M_i}^{\infty} f_{\sigma_{M_i},X_i(\sigma_{M_i})}(x,y)F_{\alpha_i t, \beta_i}(L_i-y) dy \, dx,
\end{align*}
where $F_{\alpha_i t, \beta_i}$ denotes the distribution function of a gamma distribution with parameters $\alpha_i t$ and $\beta_i$ and $f_{\sigma_{M_i},X_i(\sigma_{M_i})}$ denotes the joint density function of $(\sigma_{M_i}, X_i(\sigma_{M_i}))$ provided in \cite{Bertoin}. 

Given ${\bf W}(0)={\bf x}$ and denoting by ${\bf \sigma_{\bf{L-x}}}$ and ${\bf\sigma_{\bf{M-x}}}$ the following vectors
\begin{eqnarray*}
{\bf \sigma_{\bf{L-x}}} = \left(\sigma_{L_1-x_1}, \sigma_{L_2-x_2}, \ldots, \sigma_{L_m-x_m}\right) \quad \quad 
{\bf \sigma_{\bf{M-x}}} =\left(\sigma_{M_1-x_1}, \sigma_{M_2-x_2}, \ldots, \sigma_{M_m-x_m}\right), 
\end{eqnarray*}
the time to the next maintenance action is given by
\begin{eqnarray} \label{O} 
O &=& \sum_{k=1}^{\infty} T_k \mathbf{1}_{\left\{T_{k-1}<\min(\sigma_{\bf{M-x}})<T_k<\min(\sigma_{\bf{L-x}},Y)\right\}} \\
\nonumber
&+& \sum_{k=1}^{\infty} \min(\sigma_{\bf{L-x}}+\tau,Y+\tau,T_k) \mathbf{1}_{\left\{T_{k-1}< \min (\sigma_{\bf{M-x}}), \, T_{k-1}<\min(\sigma_{\bf{L-x}},Y)<T_k \right\}}, 
\end{eqnarray}
By {\it system renewal} we mean the maintenance time in which all the degrading components are replaced and the time to the next inspection is equal to $T$. Let $Z(t)$ be the time to the next inspection,with $0<Z(t)<T$. Starting with a new brand system at time 0, that is, $({\bf W}(0),Z(0))=({\bf 0}_m,T)$ the time to the next renewal is given by
$$R=\inf\left\{t > 0, \ ({\bf W}(t),Z(t))=({\bf 0}_m, T) \right\}. $$
To simplify, we denote by ${\bf 0}$ a vector of zeros of dimension $m$. 
After a renewal, the system is in the {\it as good as new} initial state and its future evolution does not depend any more on the past. These renewal times are regeneration points for the process describing the evolution of the maintained system.  

However, describing the system state using renewal theory is rather tricky. To deal with it, we take advantage of the semi-regenerative properties of the process considering the {\it maintenance times} as semi-regeneration points. By {\it semi-regenerative cycle} we mean the time between two successive maintenance actions.

\subsection{Evolution of the maintained system}
If $Y$ follows an exponential distribution and the degradation follows a gamma process, after each maintenance action, the system evolution depends on the system state known at this maintenance time. Let $O_k$ the time between the $(k-1)$-th maintenance action and the $k$-th maintenance action with $k=1, 2, \ldots, $ and $O_0=0$. The discrete-time process just after the maintenance time $O_k$
$$A_k=(X_{1}(O_k^+),X_{2}(O_k^+), \ldots, X_{m}(O_k^+), Z(O_k^+))$$ is a Markov Chain with continuous state space 
$$[0,M_1)\times [0,M_2) \times \ldots [0,M_m)\times (0,T]. $$

If the chain $\left\{A_i, \ i=1, 2, \ldots \right\}$  comes back to the state $({\bf 0},T)$ almost surely, that is $({\bf 0}, T)$ is a recurrent state, then $({\bf 0},T)$ is a regeneration point and it proves  (see \cite{Grall}, \cite{Huynh1} and \cite{Mercier2}) the existence of a stationary measure $\pi$ solution of the equation
\begin{equation} \label{infinitepi}
\pi(\cdot)=\int_{0}^{M_1} \int_{0}^{M_2} \ldots \int_{0}^{M_m}\int_{0}^{T}\mathbb{Q}(\cdot |({\bf x},w))\pi (d{\bf x},dw), 
\end{equation}
where $\mathbb{Q}(\cdot |({\bf x},w))$ denotes the kernel given by
\begin{eqnarray*}
\mathbb{Q}((d{\bf y},dv)| ({\bf x},w))&=&P_{({\bf x},w)}(A_1^{+} \in (d{\bf y},dv)) \\ \nonumber
&=& P(A_1^+ \in (d{\bf y},dv)|({\bf W}(0),Z(0))=({\bf x},w)).
\end{eqnarray*} 
${\bf x}=(x_1,x_2, \ldots, x_m)$ and ${\bf y}=(y_1,y_2, \ldots, y_m)$ stand for the vectors of the degradation levels. Starting with $(x_1, x_2, \ldots, x_m, w)$, with $x_i<M_i$ for $i \in I^m$ and $0 < w \leq T-\tau$, the successive inspection times $T_1, T_2, \ldots, $ up to a maintenance time are given by
\begin{equation} \label{inspectiontimes}
    T_1=w, \quad T_{k}=w+(k-1)T, \quad k=1, 2, \ldots
\end{equation}
with $T_0=0$ and time between inspections equals to
$$B_1=T_1-T_0=w, \quad B_k=T_{k+1}-T_k=T, \quad k=1, 2, \ldots. $$
Next result provides conditions to assure that the mean length of a replacement cycle is finite. 

\subsection{Expected time to a renewal cycle}
In this section, a result that assures a finite expected time to the system renewal is given.  The proof has been performed using the same reasoning as in \cite{Mercier2}. 
\begin{lemma}
If $\mu<1$, where $\mu$ is given by
\begin{equation*}
    \mu=1-\bar{F}_Y(T-\tau) \prod_{i=1}^{m} \left(\bar{F}_{\alpha_i \tau,\beta_i}(M_i) F_{\alpha_i (T-\tau),\beta_i}(L_i-M_i)\right), 
\end{equation*}
then the stationary distribution shown in Eq. (\ref{infinitepi}) exists. 
\end{lemma}

\begin{proof}
Starting with a new system at $t=0$ and $({\bf W}(0),Z(0))=({\bf 0},T))$, let $R$ be the time to the next complete replacement, 
\begin{equation*}
   R=\inf\left\{ t \geq 0, ({\bf W}(t), Z(t))=({\bf 0},T)\right\}.
\end{equation*}
Notice that these maintenance times correspond to inspection times in which all the degradation levels exceed the preventive thresholds. Let $\xi$ be the first inspection time in which a replacement is performed, 
\begin{equation*}
    \xi=\inf \left\{n \geq 1, ({\bf W}(T_n), Z(T_n))=({\bf 0},T)\right\}.
\end{equation*}
We compute $P(\xi = 1)$. It corresponds to the situation in which all the degradation levels of the degrading components exceed the preventive threshold at time $T$ and the components do not fail in $(0,T-\tau)$. For the degrading component $i$, it corresponds to 
\begin{eqnarray*}
P_i &=& P(X_i(T-\tau) \leq L_i, \, X_i(T) \geq M_i) \\ &=&P(X_i(T-\tau) \leq L_i, \, X_i(\tau)+X_i(T-\tau) \geq M_i).
\end{eqnarray*}
Hence, 
\begin{eqnarray*}
P_i &=& \int_{0}^{L_i} f_{\alpha_i(T-\tau),\beta_i}(u_i)\bar{F}_{\alpha_i \tau, \beta_i}(M_i-u_i) ~ du_i\\
& \geq &  \bar{F}_{\alpha_i \tau, \beta_i}(M_i) F_{\alpha_i(T-\tau),\beta_i}(L_i).  
\end{eqnarray*}
Therefore, 
\begin{eqnarray*}
P(\xi > 1) &=& 1-P(\xi=1) \\
&\leq & 1-\bar{F}_Y(T-\tau) \prod_{i=1}^{m} \bar{F}_{\alpha_i \tau,\beta_i}(M_i) F_{\alpha_i (T-\tau),\beta_i}(L_i) \\
&\leq & 1-\bar{F}_Y(T-\tau) \prod_{i=1}^{m} \bar{F}_{\alpha_i \tau,\beta_i}(M_i) F_{\alpha_i (T-\tau),\beta_i}(L_i-M_i).
\end{eqnarray*}
For all $n \in \mathbb{N}^*$, by the Markov property as it appears in \cite{Mercier2} we get
\begin{eqnarray*}
P(\xi > n+1) &=&  P(\xi > n+1, \xi >n ) \\
&=& \mathbb{E}\left[\mathbf{1}_{\left\{\xi > n\right\}}\mathbb{E}[\mathbf{1}_{\left\{\xi > n+1\right\}}|\mathcal{F}_{T_n}]\right] \\
&=& \mathbb{E}\left[\mathbf{1}_{\left\{\xi > n\right\}} h({\bf W}(T_n), Z(T_n))\right],
\end{eqnarray*}
with 
$$h({\bf x},v)= \mathbb{E}\left[\mathbf{1}_{\left\{ \xi > n+1  \right\}} | ({\bf W}(T_n), Z(T_n))=({\bf x},v)  \right]=\mathbf{P}_{({\bf x},v)}(\xi > 1).$$
That is, the future evolution of the system after time $T_n$ only depends on the state at time $T_n$.

For $v > \tau$, we get that
\begin{eqnarray*}
\mathbf{P}_{({\bf x},v)}(\xi >1) &=&  1-\bar{F}_Y(v-\tau) \prod_{i=1}^{m} \int_{0}^{L_i-x_i} f_{\alpha_i(v-\tau),\beta_i}(u_i)\bar{F}_{\alpha_i \tau, \beta_i}(M_i-x_i-u_i) ~ du_i\\
& \leq & 1-\bar{F}_Y(v-\tau) \prod_{i=1}^{m} \left(\bar{F}_{\alpha_i \tau,\beta_i}(M_i) F_{\alpha_i (v-\tau),\beta_i}(L_i-M_i)\right) \\
& \leq & 1-\bar{F}_Y(T-\tau) \prod_{i=1}^{m} \left(\bar{F}_{\alpha_i \tau,\beta_i}(M_i) F_{\alpha_i (T-\tau),\beta}(L_i-M_i)\right)=\mu. 
\end{eqnarray*}
For $v \leq \tau$, we get that
\begin{eqnarray*}
\mathbf{P}_{({\bf x},v)}(\xi >1) & \leq & 1-\bar{F}_Y(v) \prod_{i=1}^{m} \bar{F}_{\alpha_i v, \beta_i}(M_i-x_i) \\
& \leq & 1-\bar{F}_Y(\tau) \prod_{i=1}^{m}\bar{F}_{\alpha_i \tau, \beta_i}(M_i)  \\
    & \leq & 1-\bar{F}_Y(T-\tau) \prod_{i=1}^{m}\bar{F}_{\alpha_i \tau, \beta_i}(M_i)  F_{\alpha_i (T-\tau),\beta_i}(L_i-M_i)=\mu,  
\end{eqnarray*}
since $\tau \leq T-\tau$.
Hence, 
$$P(\xi > n+1) \leq \mathbf{E}\left[\mathbf{1}_{\left\{\xi > n\right\}} \mu \right]=\mu P(\xi > n), $$
for all $n \geq 1$, and consequently
$$P(\xi > n) \leq \mu^n, \quad \forall n \geq 1. $$
It provides that
\begin{equation*}
    \mathbb{E}\left[\xi\right]=\sum_{n=0}^{\infty} P(\xi > n) \leq \sum_{n=0}^{\infty} \mu^n < \infty. 
\end{equation*}
If $\mu<1$, then $\left\{\left(X_1(t), X_2(t), \ldots, X_m(t),Z(t)\right), t \geq 0\right\}$ is a regenerative process with bounded mean length cycle
\begin{equation*}
    \mathbb{E}\left[\sum_{k=1}^{\xi} T_k\right] \leq \mathbb{E}(\xi) T. 
\end{equation*}
\end{proof}
\subsection{Transition kernel of the semi-regenerative process}
Given $({\bf x},w)$, the kernel $\mathbb{Q}(\cdot |({\bf x},w))$ 
$$\mathbb{Q}\left(({\bf dy}, dv)|({\bf x},w)\right)=P\left(A_1^{+} \in ({\bf dy}, dv)|(W(0),Z(0))=({\bf x},w)\right), $$
is next obtained. 
As in \cite{Zhang2}, in a maintenance time, each degrading component of the system can be  into one of the following disjoint sets ($A \cup B \cup C = \Omega$, where $\Omega$ denotes the set of the degrading components). 
\begin{itemize}
\item A: components whose degradation do not exceed their preventive thresholds (``{\it healthy components}''). 
\item B: components whose degradation levels exceed their preventive thresholds but not their failure threshold (``{\it worn-out components}''). 
\item C: components whose degradation exceed the failure threshold (``{\it failed component}'').  
\end{itemize} 
To evaluate the transition kernel, we consider three cases:
\begin{itemize}
    \item {\bf Case 1.} All the degrading components are maintained in $O_1$. It means that $A=\emptyset$, $B \cup C=\Omega$. Just after the maintenance, all the degradation levels reset to zero. 
    \item {\bf Case 2.} None of the degrading components are maintained in $O_1$. It means that 
$A=\Omega$, $B=C=\emptyset$. In $O_1^+$, the degradation levels of the degrading components are  ${\bf y}=(y_1,y_2, \ldots, y_m)$, with $0 \leq x_i \leq y_i<M_i$ for $i \in I^m$. 
    \item {\bf Case 3.} Some degrading components are maintained and others are left as they were in $O_1$. It means that $A\neq \emptyset$, $B \cup C \neq \emptyset$. Let $D$ be the set of maintained degrading components, hence ${\bf y}=(y_1,y_2, \ldots, y_m)$, with $y_i=0$ for $i \in D$ and $0 \leq x_i \leq y_i <M_i$ for $i \in D^c$.
\end{itemize}
Next, we compute the probabilities associated with Case 1, Case 2 and Case 3. 
\subsubsection*{Case 1}
The degradation of the degrading components exceed the preventive thresholds in a maintenance time. That is, 
$$\left\{T_{k-1} < \min({\bf \sigma}_{{\bf M-x}})<\max(\sigma_{\bf{M-x}})<\min(\sigma_{\bf{L-x}}+\tau,Y+\tau,T_k), \quad T_{k-1} <Y\right\},  $$
for $k=1, 2, \ldots$. 
Two scenarios are envisioned for Case 1, depending if the components fail or not.  

  \begin{enumerate}
      \item None of the components fail in $(T_{k-1},T_k)$ for $k=1, 2, \ldots $. That is, 
      \begin{equation*} \label{case11}
      \left\{T_{k-1} < \min(\sigma_{\bf{M-x}})<\max(\sigma_{\bf{M-x}})<T_k<\min(\sigma_{\bf{L-x}},Y)\right\}. 
      \end{equation*}
      with probability equals to
      \begin{eqnarray*}
      \mathbb{P}_{(1,1)}^k &=& \prod_{i=1}^{m} \left(\int_{T_{k-1}}^{T_k} f_{\sigma_{M_i-x_i}}(u_i)\bar{F}_{\sigma_{L_i-x_i}-\sigma_{M_i-x_i}}(T_k-u_i) du_i \right)\bar{F}_Y(T_k) \\
      &=&  \left(\prod_{i=1}^{m} G_i(T_{k-1},T_k,T_k)\right)\bar{F}_Y(T_k), 
      \end{eqnarray*}
      where 
      $G_i(T_{k-1},T_k,T_k)$ is given by 
      \begin{eqnarray} \label{G}
    G_i(w_1,w_2,w_3) &=& \int_{w_1}^{w_2}f_{\sigma_{M_i-x_i}}(u_i) \bar{F}_{\sigma_{L_i-x_i}-\sigma_{M_i-x_i}}(w_3-u_i)~du_i, \quad w_1<w_2.  
\end{eqnarray}
\item Failure in $(T_{k-1},T_k)$, for $k=1, 2, \ldots$ Two cases are envisioned: 
      \begin{enumerate} 
      \item Failure in $ t \in (T_{k-1},T_k-\tau)$. The maintenance action is performed between inspections and the next inspection is scheduled $v=T_k-\tau-t$ units after. 
      \item Failure in $(T_k-\tau,T_k)$ The maintenance action is performed in $T_k$ and the next inspection is scheduled $v=T$ units after. 
      \end{enumerate}
We start with (a). The component fails in $ t \in (T_{k-1},T_k-\tau)$ and $max(\sigma_{M-x})<t+\tau$. It happens with the following probability
      \begin{eqnarray} \label{reduced1}
&& \mathbf{1}_{\left\{t+\tau<T_k\right\}}\int_{T_{k-1}}^{t+\tau} du_1 \int_{T_{k-1}}^{t+\tau} du_2 \ldots  \int_{T_{k-1}}^{t+\tau} du_m    \left(-\frac{d}{dt}A(t)\right), 
      \end{eqnarray}
  where
      \begin{equation} \label{eqA}
          A(t)=\left(\bar{F}_Y(t) \prod_{i=1}^{m} f_{\sigma_{M_i-x_i}}(u_i)\bar{F}_{\sigma_{L_i-x_i}-\sigma_{M_i-x_i}}(t-u_i)\right). 
      \end{equation}
 The probability of this event is equal to
      \begin{eqnarray*}
 \mathbb{P}_{(1,2)}^k(dv) = \mathbf{1}_{{\left\{v<T-\tau\right\}}} \left(\int_{T_{k-1}}^{T_k-v} du_1 \int_{T_{k-1}}^{T_k-v} du_2 \ldots \int_{T_{k-1}}^{T_k-v} du_m   \left(-\frac{d}{dv}  A(T_k-v-\tau) \right) \right),
\end{eqnarray*}
where $A(t)$ is given by Eq. (\ref{eqA}).
For scenario (b), a component fails in $(T_k-\tau,T_k)$ and all the degrading components are replaced at time $T_k$. It happens with the following probability:
\begin{eqnarray*}
\label{eq2} 
&& \mathbb{P}_{(1,3)}^k=\int_{T_{k-1}}^{T_k} du_1 \int_{T_{k-1}}^{T_k} du_2  \ldots  \int_{T_{k-1}}^{T_k} du_m  \int_{T_k-\tau}^{T_k}  \left(-\frac{d}{dv}A( v)\right) dv.  
\end{eqnarray*}
For the particular case $k=1$ and $w<\tau$, the probability that all the degrading components are replaced at time $T_1=w$ is given by 
\begin{equation*} 
\label{eq3}
\mathbb{P}^*(w)=\prod_{i=1}^{m}\bar{F}_{\alpha_i w, \beta_i}(M_i).
\end{equation*}
Given $({\bf x}, w)$ with inspection times given by Eq.(\ref{inspectiontimes}), the expression for the kernel when all the degrading components are replaced in a maintenance time is given by
\begin{eqnarray} \nonumber
    \mathbb{Q}_1((d {\bf y},dv)|({\bf x},w))&=&  \prod_{i=1}^{m}  \delta_0(dy_i)\left(\mathbb{P}_{(1,1)}^1\delta_T(dv) +\mathbb{P}_{(1,2)}^1(dv)+\mathbb{P}_{(1,3)}^1\delta_T(dv)\right)\mathbf{1}_{\left\{w \geq \tau\right\}} \\ \nonumber
    &+& 
    \prod_{i=1}^{m} \delta_0(dy_i)\sum_{k=2}^{\infty}\left(\mathbb{P}_{(1,1)}^k\delta_T(dv) +\mathbb{P}_{(1,2)}^k(dv)+\mathbb{P}_{(1,3)}^k\delta_T(dv)\right) \\ \label{Q1}
    &+& 
    \prod_{i=1}^{m} \delta_0(dy_i) \mathbb{P}^*(w) \delta_T(dv) \mathbf{1}_{\left\{w<\tau\right\}}.
\end{eqnarray}
 \end{enumerate}   
\subsubsection*{Case 2}
All the components of ${\bf y}=(y_1,y_2, \ldots, y_m)$ are strictly greater than zero in $O_1^+$. The maintenance time is triggered by the failure of the non-degrading {\color{black} part}. 
\begin{enumerate}
    \item  Non-degrading {\color{black} part fails} in $ t \in (T_{k-1},T_k-\tau)$,  Hence the time to the next inspection is equal to $T_k-\tau-t$.   
    \item Non-degrading {\color{black} part fails} in $ t \in (T_{k}-\tau,T_k)$, hence the time to the next inspection is $v=T$. 
\end{enumerate}
As in \textbf{Case 1}, $k=1$ with $w<\tau$ is a particular case. \\
Starting with $({\bf x},w)$, the probability that the non-degrading {\color{black} part fails} at time $t$ with $T_{k-1}<t<T_{k}-\tau$, and the degradation levels of the degrading components are equal to $(y_1,y_2, \ldots, y_m)$ in $t+\tau$ is equal to 
\begin{equation*}
\mathbb{P}_{(2,1)}^k(d {\bf y}, dv)=\mathbf{1}_{\left\{T_{k-1}<T_k-\tau-v<T_k-v<T_k\right\}}f_Y(T_k-\tau-v)~dv\prod_{i=1}^{m}f_{\alpha_i(T_k-v),\beta_i}(y_i-x_i)~dy_i, 
\end{equation*}
where $v=T_k-\tau-t$. 

If the non-degrading {\color{black} part fails} in $(T_k-\tau, T_k)$, maintenance intervention is performed at time $T_k$. The probability that the non-degrading {\color{black} part fails} in $(T_k-\tau, T_k)$ and the degradation level of the degrading component $i$ is $(y_i,y_i+dy_i)$, for $i \in I^m$ with $0<x_i<y_i<M_i$ at time $T_k$ is equal to
\begin{eqnarray*}\label{out1}
\mathbb{P}_{(2,2)}^k(d {\bf y})=\left(\int_{T_k-\tau}^{T_k}f_Y(z)dz\right) \prod_{i=1}^{m}f_{\alpha_iT_k, \beta_i}(y_i-x_i)~dy_i. 
\end{eqnarray*}
If $T_1=w<\tau$, then the probability that the non-degrading {\color{black} part fails} in $(0, w)$ and the degradation level of the degrading component $i$ is $(y_i,y_i+dy_i)$, for $i \in I^m$ with $0<x_i<y_i<M_i$ at time $T_1=w$ is equal to
\begin{equation*} \label{out2}
\mathbb{P}^*(d{\bf y})=F_Y(w) \prod_{i=1}^{m}f_{\alpha_i w, \beta_i}(y_i-x_i)~dy_i. 
\end{equation*}
Hence, the kernel for case 2 is equal to
\begin{eqnarray} \nonumber
     \mathbb{Q}_2((d {\bf y},dv)|{({\bf x},w))} &=&  \mathbb{P}^*(d{\bf y}) \mathbf{1}_{\left\{w < \tau \right\}}\delta_T(dv) + \left(\mathbb{P}_{(2,1)}^1(d{\bf y},dv)+\mathbb{P}_{(2,2)}^1(d{\bf y})\delta_T(dv)\right)  \mathbf{1}_{\left\{w \geq \tau \right\}}\\ \label{Q2}
     &+& \sum_{k=2}^{\infty} \left( \mathbb{P}_{(2,1)}^k(d {\bf  y},dv)+\mathbb{P}_{(2,2)}^k(d{\bf y})\delta_T(dv)\right). 
\end{eqnarray}
\subsubsection*{Case 3}
Some degrading components are replaced and the rest are left as they are in $O_1^+$. That is, $A \neq \emptyset$, $B \cup C \neq \emptyset$. It means that some components of vector {\bf y} are strictly greater than zero and the rest are equal to zero. 
    Let $D$ be the set of indexes of the vector {\bf y} equals to zero. As in \textbf{Case 1}, different scenarios are envisioned:
\begin{enumerate}
    \item None of the degrading components fail in $(T_{k-1},T_k)$ but a maintenance action is performed in $T_k$.
    \item {\color{black} Failure} in $(T_{k-1},T_k)$. 
    \begin{itemize}
        \item {\color{black} Failure} in $t \in (T_{k-1},T_k-\tau)$, then the maintenance action is performed in $t+\tau$ and the next inspection is scheduled $T_k-t-\tau$ time units after. 
        \item {\color{black} Failure} in $t \in (T_k-\tau, T_k)$. The maintenance action is performed in $T_k$ and the next inspection is scheduled $T$ time units after. 
    \end{itemize}
\end{enumerate} 
We start with case 1. If none of the components fail in $(T_{k-1},T_k)$, the probability that some degrading components are replaced in an inspection time and the rest are left as they are and none failure happens between inspections is equal to 
\begin{eqnarray*}
\mathbb{P}_{(3,1)}^k(d {\bf y}) &=& \bar{F}_Y(T_k)  \prod_{i \in D} \left(\int_{T_{k-1}}^{T_k}f_{\sigma_{M_i-x_i}}(u_i)\bar{F}_{\sigma_{L_i-x_i}-\sigma_{M_i-x_i}}(T_k-u_i) du_i \right) \prod_{j \in D^c}{f}_{\alpha_j T_k, \beta_j}(y_j-x_j)~dy_j \\
&=& \bar{F}_Y(T_k)  \prod_{i \in D} G_{i}(T_{k-1},T_k,T_k)\prod_{j \in D^c}{f}_{\alpha_j T_k, \beta_j}(y_j-x_j)~dy_j,   
\end{eqnarray*}
with $x_j<y_j$ for $j \in D^c$. 

We start with the first case. The probability that a {\color{black}degrading} component ({\color{black} or the non-degrading part}) fails in $(t,t+dt)$ with $T_{k-1}<t<T_k-\tau$ and some preventive thresholds are exceed at the maintenance intervention time $t+\tau$ is given by
\begin{eqnarray*}
&& { {\bf 1}_{\lbrace T_{k-1} < t < T_{k}-\tau \rbrace}} \int_{T_{k-1}}^{t+\tau} du_1 \int_{T_{k-1}}^{t+\tau} du_2  \ldots \int_{T_{k-1}}^{t+\tau} du_m \frac{-d}{dt} B(t)   \prod_{j \in D^c}{f}_{\alpha_j(t+\tau), \beta_j}(y_j-x_j)dy_j,  
\end{eqnarray*}
where
\begin{equation*}
    B(t)=\bar{F}_Y(t)\prod_{i \in D}^{} f_{\sigma_{M_i-x_i}}(u_i)\bar{F}_{\sigma_{L_i-x_i}-\sigma_{M_i-x_i}}(t-u_i) du_i. 
\end{equation*}
In this case, next inspection is scheduled $T_k-\tau-t$ time units after. Denoting by $v=T_k-\tau-t$, the probability that some degrading components are replaced in a maintenance time and the time to the next inspection time is $(v,v+dv)$ is equal to
\begin{eqnarray*}
&& \mathbb{P}_{(3,2)}^k(dv, d {\bf y})= \int_{T_{k-1}}^{T_k-v} du_1 \int_{T_{k-1}}^{T_k-v} du_2 \ldots  \int_{T_{k-1}}^{T_k-v} du_m \frac{-d}{dv} B(T_k-v-\tau) \prod_{j \in D^c} {f}_{\alpha_j(T_k-v), \beta_j}(y_j-x_j)  ~dy_j 
\end{eqnarray*}
for $y_j>x_j$, for $j \in A^c$. If it fails in $(T_k-\tau,T_k)$, maintenance intervention is performed at time $T_k$. The probability that a {\color{black}degrading component (or the non-degrading part)} fails in $(T_k-\tau,T_k)$ and some degrading components are replaced at time $T_k$ is given by
\begin{eqnarray*}
&& \mathbb{P}_{(3,3)}^k(d {\bf y})=  \int_{T_{k-1}}^{T_k} du_1 \int_{T_{k-1}}^{T_k} du_2 \ldots \int_{T_{k-1}}^{T_k} du_m
\int_{v=T_k-\tau}^{T_k} -\frac{d}{dv} B(v) ~dv \prod_{j \in D^c} {f}_{\alpha_j T_k, \beta_j}(y_j-x_j) ~dy_j,  
\end{eqnarray*}
for $y_j>x_j$, for $j \in A^c$. In the particular case that $w<\tau$, we get that the probability that some degrading components are replaced and the rest are left as they are is equal to
$$\mathbb{P}^*(d{\bf y})=\prod_{i \in D} \bar{F}_{\alpha_i w, \beta_i}(M_i) \prod_{j \in D^c} f_{\alpha_jw, \beta_j}(y_j-x_j) {dy_j}. $$
Finally, 
\begin{eqnarray} \nonumber
    \mathbb{Q}_3((d {\bf y}, dv)|{ ({\bf x},w))} &=& \prod_{i \in D} \delta_0(dy_i) \sum_{k=2}^{\infty}\left(\mathbb{P}_{(3,1)}^k(d {\bf y}) \delta_T(dv)+\mathbb{P}_{(3,2)}^k(d {\bf y},dv)+\mathbb{P}_{(3,3)}^k(d {\bf y})\delta_T(dv)\right) \\ \nonumber
    &+& \prod_{i \in D} \delta_0(dy_i) \left(\mathbb{P}_{(3,1)}^1(d {\bf y}) \delta_T(dv)+\mathbb{P}_{(3,2)}^1(d {\bf y},dv)+\mathbb{P}_{(3,3)}^1(d {\bf y})\delta_T(dv)\right)\mathbf{1}_{\left\{w \geq\tau\right\}} \\ &+& \mathbb{P}^*(d{\bf y}) \delta_T(dv)\mathbf{1}_{\left\{w < \tau\right\}}. \label{Q3}
\end{eqnarray}
\begin{proposition}
The transition kernel of the semi-regenerative process $({\bf W}(t), Z(t))$ is given by
$$\mathbb{Q}((d {\bf y}, dv)|({\bf x},w))=\mathbb{Q}_1((d {\bf y}, dv)|({\bf x},w))+\mathbb{Q}_2((d {\bf y}, dv)|({\bf x},w))+\mathbb{Q}_3((d {\bf y}, dv)|({\bf x},w)), $$
\end{proposition}
where $\mathbb{Q}_1$, $\mathbb{Q}_2$ and $\mathbb{Q}_3$ are given by Eqs. (\ref{Q1}), (\ref{Q2}) and (\ref{Q3}) respectively. 
\section{The objective cost function} \label{probability}
The semi-regenerative properties of the maintained system allow for studying the asymptotic
behaviour of the expected cost focused on a semi-regenerative cycle. We remind that by semi-regenerative cycle we mean the time between two maintenance actions. Let $C(t)$ be the cost of the system at time $t$. 
Using Proposition 4.1. in \cite{Berenguer}, for a semi-regenerative process with an unique stationary probability distribution $\pi$ (Eq. \ref{infinitepi}), 
\begin{equation} \label{cinfinito}
C_{\infty}=\lim_{t \longrightarrow \infty} \frac{\mathbb{E}[C(t)]}{t}= \frac{\mathbb{E}_{\pi}[C(O_1)]}{\mathbb{E}_{\pi}[O_1]}, 
\end{equation}
where $O_1 \equiv O$  stands for the time to a maintenance time given by Eq. (\ref{O}).

Developing (\ref{cinfinito}):
\begin{eqnarray} \label{costfunction}
C_{\infty}(T, M_1, M_2, \ldots, M_m) &=& \frac{E_{\pi}[C^p(O_1)]}{E_{\pi}[O_1]}
+ \frac{E_{\pi}[C^c(O_1)]}{E_{\pi}[O_1]} \\ \nonumber &+&  \frac{E_{\pi}[C^{nm}(O_1)]}{E_{\pi}[O_1]}+\frac{E_{\pi}[C(I(O_1))]}{E_{\pi}[O_1]}+\frac{E_{\pi}[D(O_1)]}{E_{\pi}[O_1]}-\frac{E_{\pi}[R(O_1)]}{E_{\pi}[O_1]}, 
\end{eqnarray}
where, given $\pi$, $E_{\pi}[C^p(O_1)]$ and $E_{\pi}[C^c(O_1)]$  stand for the expected costs due to the preventive and corrective maintenance of the degrading components in a semi-regenerative cycle,  $E_{\pi}[C^{nm}(O_1)]$ denotes the expected cost due to the corrective replacements of the non-degrading {\color{black} part} in a semi-regenerative cycle, $E_{\pi}[C(I(O_1))]$ corresponds to the expected cost due to inspections, $E_{\pi}[C(D(O_1))]$ the expected cost due to the downtime, $E_{\pi}[R(O_1)]$ stands for the expected reward obtained in a semi-regenerative cycle and $E_{\pi}[O_1]$ denotes the expected length of a semi-regenerative cycle. 
All the terms involved in Eq. (\ref{costfunction}) are developed in the { Appendix} of this paper. 

The unconstrained optimization problem is to find the values $(T, M_1, M_2, \ldots, M_m)$ that minimize $C_{\infty}(\cdot)$ given in Eq. (\ref{costfunction}). That is, finding the values $(T_{opt}, M_{1,opt}, M_{2,opt}, \ldots, M_{m,opt})$ such that
$$C(T_{opt}, M_{1,opt}, M_{2,opt}, \ldots, M_{m,opt})=\inf\left\{C_{\infty}(T, M_1, M_2, \ldots, M_m),  \quad T > 2\tau, \quad M_i \leq L_i, \quad i \in I^m\right\}. $$

We assume that a safety requirement that bounds the probability of a critical situation is imposed. 
{\it A critical situation} happens in this system when all the degrading components fail in a semi-regenerative cycle. If all the degrading components correspond to the mechanical part of the system   and they are put in operation under a parallel configuration, a critical situation corresponds to the failure of the whole mechanical. A real life example of this parallel configuration is found in a gas company that pumps and distributes gas. Critical components are subject to degradation and, to ensure a high availability of the pumps, critical components are placed in a parallel setting \cite{Keizer}. 
 
We can approach the search of the optimal maintenance strategy {\color{black} bounding} the probability of a {\it critical situation} for the system. {\color{black} By {\it critical situation} we mean that all the degrading components fail between two successive maintenance actions}. The safety constraint requires that the probability that all the degrading components fail between two successive maintenance actions is bounded by a value $w$ where $w$ represents the safety probability limit \cite{Aven} with $0<w<1$.

Given $({\bf x},v)$, a critical situation {\color{black} between $T_{k-1}$ and $T_k$} corresponds to the following event
\begin{equation} \label{eqrestriccion}
C_{x,v}=\left\{T_{k-1} < \min(\sigma_{{\bf M-x}})<\min(\sigma_{{\bf L-x}},Y)<\max(\sigma_{{\bf L-x}})<\min(\sigma_{{\bf L-x}}+\tau,Y+\tau,T_k)\right\}, 
\end{equation}
for $k=1, 2, \ldots$. Hence, the probability of a critical situation in a semiregenerative cycle is given by
\begin{equation} \label{criticalsituation}
P_{\pi}(T,M_1,M_2, \ldots, M_m) = \int \pi(dx, dv) \sum_{k=1}^{\infty}P(C_{x,v}). 
\end{equation}
Figure \ref{probcritica} shows the probability of a critical situation with a time between inspections equals to $T=100$ time units and parameters 
$M_i=3$, $L_i=6$, $\lambda=0.025$, $\beta_i=1$ for $i \in I^m$, 
and $\alpha_1=0.2$, $\alpha_2=0.3$, $\alpha_3=0.4$, $\alpha_4=0.5$, and  $\alpha_5=0.6$ versus $\tau$. As we can check in Figure \ref{probcritica}, the probability of a critical situation is increasing in $\tau$ and decreasing with the number of degrading components $m$. 

\begin{figure}[H]
\begin{center} 
 \includegraphics[scale=0.55]{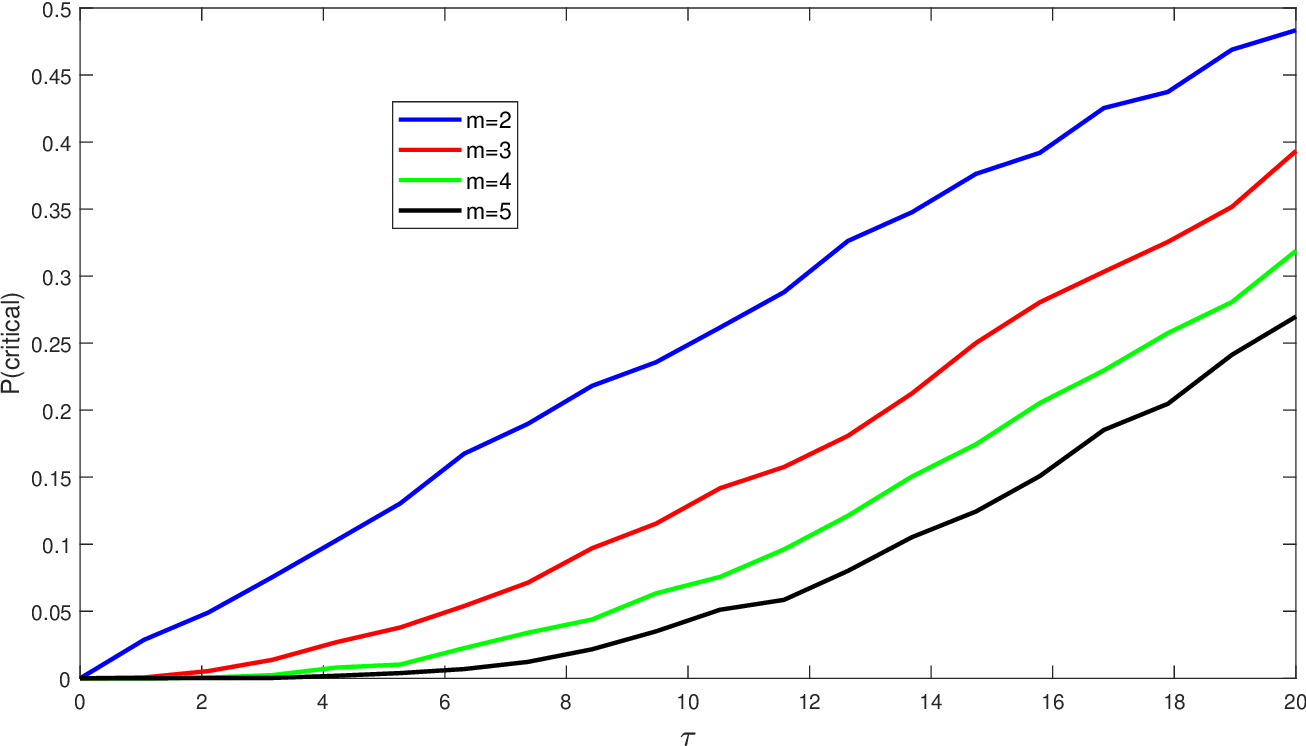}
  \caption{Probability of a critical situation versus $\tau$.}
  \label{probcritica}
 \end{center} 
\end{figure}

The constrained optimization problem is to find the values $(T, M_1, M_2, \ldots, M_m)$ that minimize $C_{\infty}(\cdot)$ given in Eq. (\ref{costfunction}) under the safety constraint. That is, finding the values $(T_{opt}, M_{1,opt}, M_{2,opt}, \ldots, M_{m,opt})$ such that
\begin{equation} \label{constrainedproblem}
C(T_{opt}, M_{1,opt}, M_{2,opt}, \ldots, M_{m,opt})=\inf\left\{C_{\infty}(T, M_1, M_2, \ldots, M_m),  \quad (T, M_1, M_2, \ldots, M_m) \in \mathbb{D}\right\},  
\end{equation}
where $\mathbb{D}$ is the set of vectors $(T, M_1, M_2, \ldots, M_m)$ that fulfills Eq. (\ref{criticalsituation}). That is, 
\begin{equation*}
\mathbb{D}=\left\{(T, M_1, M_2, \ldots, M_m); \quad T > 2 \tau,  \, M_i \leq L_i, \, P_{\pi}(T,M_1, \ldots, M_m) \leq w \right\}
\end{equation*}
Due to the complexity of the probability given by Eq.(\ref{criticalsituation}), there is no hope here to find analytic conditions that could ensure the monotony of Eq. (\ref{criticalsituation}) with respect to $T$ or with respect to $(M_1, M_2, \ldots, M_m)$. Hence, resolution of the constrained optimization problem given in Eq. (\ref{constrainedproblem}) requires the use of numerical methods. Numerical examples of the constrained optimization problem are given in Section \ref{numerical}. Below, the analytic development of the terms included in Eq. (\ref{costfunction}) are explained.

\subsection{Expected reward in a semi-regenerative cycle}
Each working degrading component provides a reward. This reward depends on the component deterioration: when the deterioration level of this degrading component increases the reward decreases. A classical exponential reward function given in \cite{Niese} is used in this paper. {\color{black} For other reward functions see for instance \cite{Mercier} and \cite{Xiang}}.

The reward function $g_i$ for the component $i$, given its deterioration level $x$, is equal to
\begin{equation} \label{reward}
    g_i(x)=\theta_0+h\exp{(-\gamma_i x)}, \quad  \quad i=1,2, \ldots,m, 
\end{equation}
where $\theta_0, h, \gamma_i\geq 0$. The reward function $g(\cdot)$ given in Eq. (\ref{reward}) is decreasing in $x$, with maximum $\theta_0+g$ and minimum $\theta_0$ and the expected reward rate in an interval $[0,T]$ for a process $\left\{X(t), \, t \geq 0\right\}$ is given by
\begin{equation} \label{rewardrate}
 \frac{\displaystyle{\int_{0}^T} \mathbb{E}\left[g(X(t))\right]dt}{T}, \quad T > 0. 
\end{equation}

%
%

For different gamma processes with parameters $(\alpha_i,\beta_i)$ with $i \in I^m$, the expected reward rate in $[0,T]$ given by Eq. (\ref{rewardrate}) can be comparable for $\gamma_i=\gamma$ and $i \in I^m$. Since function $-g$ given in (\ref{reward}) is increasing and concave, if 
\begin{equation*}
\alpha_1 \leq \alpha_2 \leq \alpha_3 \ldots \leq \alpha_m, \quad  \beta_1 \leq \beta_2 \leq \beta_3 \ldots \leq \beta_m, \quad \alpha_1/\beta_1 \leq \alpha_2/\beta_2 \leq \ldots \leq \alpha_m/\beta_m, 
\end{equation*}
then the variables $X_1(s), X_2(s), \ldots, X_m(s)$ can be ordered under the increasing convex order using a result given by \cite{Muller}, that is, 
$$X_1(s) \prec_{icv} X_2(s) \prec_{icv} \ldots \prec_{icv} X_m(s), \quad \forall s.  $$
Using the definition of increasing convex order, we get that
\begin{equation}
\frac{\displaystyle{\int_{0}^{T}\mathbb{E}[-g_1(X_1(s))]ds}}{T} \leq \frac{\displaystyle{\int_{0}^{T}\mathbb{E}[-g_2(X_2(s))]ds}}{T}  \leq \ldots \leq \frac{\displaystyle{\int_{0}^{T}\mathbb{E}[-g_m(X_m(s))]ds}}{T},   
\end{equation}
hence, 
\begin{equation}
\frac{\displaystyle{\int_{0}^{T}\mathbb{E}[g_m(X_m(s))]ds}}{T} \leq \frac{\displaystyle{\int_{0}^{T}\mathbb{E}[g_{m-1}(X_{m-1}(s))]ds}}{T}  \leq \ldots \leq \frac{\displaystyle{\int_{0}^{T}\mathbb{E}[g_1(X_1(s))]ds}}{T},   
\end{equation}
for $T>0$. Figure \ref{fig:expectation} shows the expected reward per unit time given in Eq. (\ref{rewardrate}) for five gamma processes with parameters $\alpha_i=1+(i-1)*0.1, \quad \beta_i=1, \quad i=1, 2, 3, 4, 5$  versus $T$. Furthermore, the parameters of function $g$ shown in (\ref{reward}) are given by $\gamma_i=2$, $\theta=2$ and $h=2$. As we can check, the expected  reward rate decreases when the parameter $\alpha$ increases. 

\begin{figure}[H]
\begin{center} 
 \includegraphics[scale=0.55]{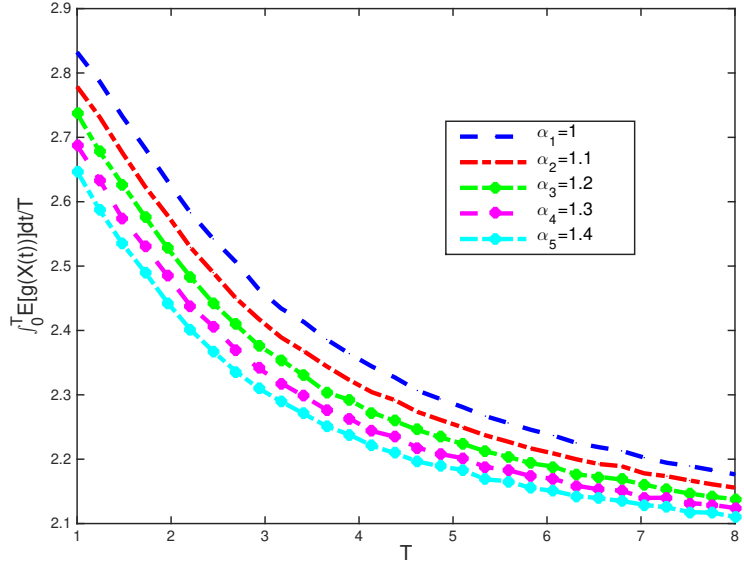}
  \caption{Expected reward rate versus $T$ for different values of $\alpha$.}
  \label{fig:expectation}
 \end{center} 
\end{figure}

The expected reward in a semi-regenerative cycle is
\begin{eqnarray*}
 \mathbb{E}_{\pi}\left[R(O_1)\right] &=&  \int_{}^{} \pi(d {\bf x},dw) \mathbb{E}_{O_1, \sigma_{L_i}}\left(\sum_{i=1}^{m} \int_{0}^{\min(O_1, \sigma_{L_i})} \mathbb{E}[g_i\left(X_i(s)+x_i\right)]~  ds\right), 
\end{eqnarray*}


\subsection{Expected cost due to the corrective maintenance of the degrading components in a semi-regenerative cycle}\label{4.2}
 Given $({\bf x},w)$, the probability that the degrading component $i$ fails in $(T_{k-1},T_k)$ is given by
 $$P_i^c(T_k)=P(T_{k-1}<\min({\bf \sigma_{M-x}}), \quad T_{k-1}<\min({\bf \sigma_{L-x}},Y)< T_k, \quad \sigma_{L_i-x_i} \leq \min(T_k,{\bf \sigma_{L-x}}+\tau,Y+\tau)). $$
 Hence, the expected cost due to corrective maintenance of the degrading components in a semi-regenerative cycle $\mathbb{E}_{\pi}\left[C^c(O_1)\right]$ is given by
\begin{equation*}
    \mathbb{E}_{\pi}\left[C^c(O_1)\right]=\int_{}^{} \pi(d {\bf x},dw) \sum_{i=1}^{m} \sum_{k=1}^{\infty} C_i^c P_i^c(T_k),
\end{equation*}
where $C_i^c$ is the corrective maintenance cost of component $i$ and $P_i^c(T_k)$ is obtained in the {\bf Appendix}.  

\subsection{Expected cost due to the corrective maintenance of the non-degrading {\color{black} part} in a semi- regenerative cycle}\label{4.3}
 Given $({\bf x},w)$, the probability that the non-degrading {\color{black} part fails} in $(T_{k-1},T_k)$ is equal to
 \begin{eqnarray*}
 P^f(T_k)=P(T_{k-1}<\min({\bf \sigma_{M-x}}), \quad T_{k-1}<\min(\sigma_{\bf{L-x}},Y)<T_k, \quad Y \leq \min(T_k,\sigma_{\bf{L-x}}+\tau)), 
 \end{eqnarray*}
hence the expected cost due to these corrective maintenance actions is given by
\begin{equation*}
    \mathbb{E}_{\pi}[C^{nm}(O_1)]=\int_{}^{} \pi(d {\bf x},dw) \sum_{k=1}^{\infty} C^f P^f(T_k),
\end{equation*}
where $C^f$ is the corrective maintenance cost of the non-degrading {\color{black} part} and $P^f(T_k)$ is calculated in the {\bf Appendix}.

\subsection{Expected cost due to preventive maintenance of the degrading components in a semi-regenerative cycle}\label{4.4}
Given $({\bf x},w)$, a preventive maintenance of the degrading component $i$ is performed in $(T_{k-1},T_k)$  if the following event occurs
$$
\left\{ T_{k-1}<\min({ \sigma_{\bf{M-x}}},Y),  \hspace{0.2cm} T_{k-1}<\sigma_{M_i-x_i}<\min(\sigma_{\bf{L-x}}+\tau,Y+\tau,T_k)<\sigma_{L_i-x_i}  \right\}, 
$$
with probability
$$P_i^p(T_k)=P(T_{k-1}<\min(\sigma_{\bf{M-x}},Y), \hspace{0.2cm} T_{k-1}<\sigma_{M_i-x_i}<\min(\sigma_{\bf{L-x}}+\tau,Y+\tau,T_k)<\sigma_{L_i-x_i}).  $$
Hence, the expected cost due to preventive {\color{black} replacements} of the degrading components in a semi-regenerative cycle is given by
\begin{equation*}
\mathbb{E}_{\pi}[C^p(O_1)]=\int_{}^{} \pi(d{\bf x}, dw) \sum_{i=1}^{\infty}\sum_{k=1}^{\infty} C_i^p P_i^p(T_k),
\end{equation*}
where $C_i^p$ is the preventive maintenance cost of component $i$ and $P_i^p(T_k)$ is calculated in the {\bf Appendix}.

\subsection{Expected cost due to inspections of the system in a semi-regenerative cycle}
Given $({\bf x},w)$, the expected cost due to inspections is given by
\begin{eqnarray*}
\mathbb{E}_{({\bf x},w)}\left[C(I(O_1))\right] 
&=& C^I \sum_{k=1}^{\infty} P(\min(\sigma_{\bf{M-x}},Y)>T_k) \\
&=& C^I \sum_{k=1}^{\infty} \prod_{i=1}^{m} \bar{F}_{\sigma_{M_i-x_i}}(T_k)\bar{F}_Y(T_k).
\end{eqnarray*}
Hence, the expected cost due to inspections in a semi-regenerative cycle is given by
\begin{equation*}
\mathbb{E}_{\pi}[C(I(O_1))]=C^I \int_{}^{} \pi(d{\bf x}, dw) \sum_{k=1}^{\infty} \prod_{i=1}^{m} \bar{F}_{\sigma_{M_i-x_i}}(T_k)\bar{F}_Y(T_k),
\end{equation*}
where $C^I$ denotes the cost of an inspection.

\subsection{Expected cost due to downtime in a semi-regenerative cycle} \label{4.5}

Given $({\bf x},w)$, a downtime of a {\color{black} degrading} component in $(T_{k-1},T_k)$ takes place if the following event occurs
$$
\lbrace T_{k-1}<\min(\sigma_{\bf{M-x}}), \quad T_{k-1}<\min(Y,\sigma_{\bf{L-x}}) <T_k \rbrace.
$$

We define the following expectations for the expected downtime of non-degrading {\color{black} part}
$$
D^{nm}(T_k)=\mathbb{E}\left[(\min(Y+\tau, {\bf \sigma_{L-x}}+\tau, T_k)-Y )\mathbf{1}_{\left\{T_{k-1}<\min(\sigma_{\bf{M-x}}), \ T_{k-1}<Y <\min(\sigma_{{\bf L-x}}+\tau,Y+\tau,T_k)\right\}}\right], 
$$
and for the expected downtime of degrading component $i$
$$
D_i(T_k)=\mathbb{E}\left[(\min(Y+\tau, {\bf \sigma_{L-x}}+\tau, T_k)-\sigma_{L_i-x_i} )\mathbf{1}_{\left\{ T_{k-1}<\min(\sigma_{\bf{M-x}}), \ T_{k-1}\leq \sigma_{L_i-x_i}<\min(\sigma_{{\bf L-x}}+\tau,Y+\tau,T_k) \right\}}\right], 
$$
for $i \in I^m$. Given $({\bf x},w)$, the expected cost due to downtime is given by
\begin{eqnarray*}
 \mathbb{E}_{({\bf x},w)}\left[D(O_1)\right] =  \sum_{k=1}^{\infty} c^{nm}D^{nm}(T_k) + \sum_{k=1}^{\infty} \sum_{i=1}^{m } c_i D_i(T_k),
\end{eqnarray*}
where $c_{nm}$ is the downtime cost of non-degrading {\color{black} part}, $c_i$ is the downtime cost for degrading component $i$ with $i=1,2, \ldots, m$. 
Hence, the expected downtime in a semi-regenerative cycle is given by
\begin{equation*}
    \mathbb{E}_{\pi}\left[D(O_1)\right] = \int_{}^{} \pi(d{\bf x}, dw) \mathbb{E}_{({\bf x},w)}\left[D(O_1)\right]. 
\end{equation*}
Expectations $D^{nm}(T_k)$ and $D_{i}(T_k)$ are calculated in the {\bf Appendix}.
\subsection{Expected length of a semi-regenerative cycle} \label{4.6}
Finally given $({\bf x},w)$, the expected length of a semi-regeneration cycle is equal to
\begin{eqnarray*}
\mathbb{E}_{({\bf x},w)}[O_1] &=& \sum_{k=1}^{\infty} \mathbb{E} \left[T_k \mathbf{1}_{\left\{T_{k-1}<\min(\sigma_{\bf{M-x}})<T_k<\min(\sigma_{\bf{L-x}},Y)\right\}}\right] \\
&+& \sum_{k=1}^{\infty} \mathbb{E} \left[\min(\sigma_{\bf{L-x}}+\tau, Y+\tau, T_k) \mathbf{1}_{\left\{T_{k-1}<\min(\sigma_{\bf{M-x}}), \hspace{0.2cm} T_{k-1}< \min(\sigma_{\bf{L-x}},Y)<T_k\right\}}\right], 
\end{eqnarray*}
therefore, 
\begin{equation*}
    \mathbb{E}_{\pi}[O_1]=\int_{} \pi(d {\bf x},dw) \mathbb{E}_{({\bf x},w)}[O_1]. 
\end{equation*}
The development of this expectation is given in the {\bf Appendix}.

\section{Numerical examples} \label{numerical}
To solve the optimization problem given in (\ref{constrainedproblem}), Monte Carlo simulation method and meta-heuristic algorithms are applied \cite{GA}.

\subsection{Identical components}\label{section_identical}
We consider $m$ identical degrading components whose degradation follows a gamma process with parameters $\alpha_i=1.25$ and $\beta_i=0.5$, for all $i \in I^m$. A degrading component fails when its degradation exceeds $L=6$. Non-degrading {\color{black} part fails} following an exponential distribution with parameter $\lambda=0.025$ failures per time unit. Delay time is equal to 
$\tau=0.5$. The following sequence of costs is imposed: 
$$C_{i}^c=80 \, m.u., \, C_{i}^p=30 \, m.u., \,  C^f=80 \, m.u., \, C^I=10 \,  m.u., \, c_i=5 \, m.u. \text{per} \, t.u., \,  c^{nm}=5 \, m.u. \text{per} \,  t.u.$$
For the reward function, the following parameters are used
$$g=2, \quad \theta_0=2, \quad \gamma_i=20. $$
The safety constraint given by (\ref{constrainedproblem}) is imposed with $w=0.05$. 
To optimize the objective cost function, first, Monte Carlo simulation is used to search potential solutions for $T$ and $M$. Using these potential solutions as initial points, meta-heuristic algorithms are applied to further improve the optimal policy so that the expected maintenance cost is minimized \cite{Song}. A {\em Pattern Search} algorithm under the constrained optimization problem given by Eq. (\ref{constrainedproblem}) is used.

Table \ref{computation} shows the initial points obtained with Monte Carlo simulation (a) and the optimal values for $T$ and $M$ considering the {\em Pattern Search} algorithm for different values of $m$ (b), under the constrained optimization problem given in (\ref{constrainedproblem}). Table \ref{computation} shows that, when the number of degrading components increases, the values of $T_{opt}$ decrease and the values of $M_{opt}$ increase. And, the expected cost rate increases when the number of degrading components increases. 

\subsection{Non identical components}
Non identical degrading components is next analysed. For that, we assume that
\begin{equation}
\alpha_1=1.1, \quad \alpha_i=\alpha_{i-1}+0.1, \quad  i \in I^m, 
\end{equation}
and the rest of the parameters and costs of Section \ref{section_identical} are used. 
As before, a degrading component fails when its degradation exceeds $L_i=6$,  for $\in I^m$. Time between not-degrading {\color{black} part} failures is exponential with $\lambda=0.025$ failures per unit time. Delay time is $\tau=0.5$.

The search of the optimal maintenance policy corresponds to find the vector 
$(T, M_1, \ldots, M_m)$ that optimizes the function $C_{\infty}$ given by Eq. (\ref{costfunction}). That is
$$C_{\infty}(T_{opt},M_{1,opt}, M_{2,opt}, \ldots, M_{m,opt})=\inf\left\{C_{\infty}(T,M_1, M_2, \ldots, M_m); \, T >2 \tau, M_i \leq L \quad i \in I^m\right\}. $$
Table \ref{computation2} shows the optimal values for the time between inspections, the optimal preventive thresholds, the expected cost rate and the probability of a critical situation for different values of $m$ obtained with a {\it Genetic Algorithm}. The optimal values for $T$ decrease with the number of components, and the optimal preventive thresholds decrease when the shape parameter increases.  By adding more components, the expected cost rate increases.


\begin{table}
\centering
\subfloat[]{
\begin{tabular}{c||c||c}
m & $T_{0}$ & $M_0$ \\ \hline \hline
2 & 4.39 & 2.95  \\
3 & 4.28 & 3.17 \\
4 & 3.13 & 3.38 \\
5 & 2.47 & 3.78 \\
6 & 2.29 & 3.93   \\
7 & 1.89 & 4.18  \\
8 & 1.64 & 4.37  \\
9 & 1.53 & 4.65  \\
10 & 1.48 & 4.70 
\end{tabular}}
\qquad\qquad\qquad\qquad
\subfloat[]{
\begin{tabular}{c||c||c || c}
m & $(T_{opt},M_{opt})$ & $C_{\infty}(T_{opt},M_{opt})$ & $P_{\pi}(T_{opt},M_{opt})$ \\ \hline \hline
2 & (4.317, 3.075) & 8.140 & 0.001  \\
3 & (3.726, 3.345) & 9.649 & $<10^{-5}$\\
4 & (3.130, 3.680) & 11.657 & $<10^{-5}$ \\
5 & (2.657, 3.920) & 12.531 & $<10^{-5}$ \\
6 & (2.106, 4.215) & 13.923 & $<10^{-5}$  \\
7 & (1.918, 4.402) & 14.972 & $<10^{-5}$ \\
8 & (1.828, 4.451) & 15.795 & $<10^{-5}$ \\
9 & (1.733, 4.554) & 16.671 & $<10^{-5}$  \\
10 & (1.558, 4.662) &  17.600  & $<10^{-5}$
\end{tabular}
}
\caption{(a) Starting points. \hspace{2cm} (b) $T_{opt}, M_{opt}$ and $C_{\infty}$ using the {\it Pattern Search} method.} \label{computation}
\end{table}




\hspace{2cm}

\begin{table}[tbph]
\begin{center}
\begin{tabular}{c||c||c||c || c}
m & $T_{opt}$ & $(M_{1,opt}, M_{2,opt}, \ldots, M_{m,opt})$ & $C_{\infty}(T,M_{1,opt}, \ldots, M_{m,opt})$ & $P_{\pi}(T,M_{1,opt}, \ldots, M_{m,opt})$  \\ \hline \hline
2 & 5.998 & (2.679, 2.014) & 4.968 & 0.015  \\
3 & 3.997 & (3.392, 3.295, 3.086) & 8.237 & $5e^{-5}$ \\
4 & 2.798 & (4.318, 3.958, 3.496, 3.236) & 11.681 & $<10^{-6}$ \\
5 & 2.200 & (3.939, 3.860, 3.805, 3.724, 3.660) & 13.217 & $<10^{-6}$ 
\end{tabular}
\caption{Optimal parameter values for non-identical degrading components, optimal expected cost rate and probability of a critical situation with $\tau=0.5$,  obtained with {\it  Genetic Algorithm}.} \label{computation2}
\end{center}
\end{table}

\subsection{Sensitivity analysis}

%
%
%


A sensitivity analysis of the shape and scale parameters of the gamma process, as well as the rate parameter of the exponential distribution is performed, considering that all the degrading components are independent and identical. To carry out this analysis, the dataset of Section \ref{section_identical} is used. 

Tables \ref{computation5}, \ref{computation4} and \ref{computation6} show the relative variation for the parameter values $\alpha$ (shape), $\beta$ (scale),  and $\lambda$ (exponential rate), respectively, and for each value of $m$ from $2$ up to $10$ components. A shaded grey scale is used in the sensitivity analysis. The clearest coloured cells are the ones with the least variability, and vice versa.

 A relative measure is defined as:
$$
V_{\alpha,\beta,\lambda}= \frac{|C_{\infty}(T_{opt},M_{opt})-C_{\alpha, \beta, \lambda}|}{C_{\infty}(T_{opt},M_{opt})},
$$
where $C_{\alpha, \beta, \lambda}$ is the minimal expected cost rate obtained by varying one of the three parameters, keeping the other two fixed,  and $C_{\infty}(T_{opt},M_{opt})$ is the cost obtained in Table \ref{computation2} (b). $V_{\alpha,\beta, \lambda}$ measures the relative difference between the optimal cost and the optimal cost calculated by using the modified parameters. The closer to zero, the less influence the modified parameter values have on the solution. 
Regarding the results, parameters $\alpha$ and $\beta$ have similar effects on $V_{\alpha,\beta,\lambda}$. However, by modifying $\pm 0.05$ around $\alpha$ and $\beta$, the relative  variation is small (less than 5\%). In the case of parameter $\lambda$, almost all the values obtained are below $2\%$ variation, so it has less influence on the solution. 

\vspace{2cm}

\begin{table}[tbph]
\begin{center}
\begin{tabular}{|c|c|c|c|c|c|c|c|}
\hline
\backslashbox{$m$}{$\alpha$} & 1.10 & 1.15 & 1.20 & 1.25 & 1.30 &  1.35  & 1.40 \\
\hline
2  & \cellcolor{gris5} 0.1246 & \cellcolor{gris3} 0.0652 & \cellcolor{gris1} 0.0158 & \cellcolor{gris0} 0.0000 &\cellcolor{gris1} 0.0157 &  \cellcolor{gris2} 0.0502 & \cellcolor{gris5} 0.1374  \\

3  & \cellcolor{gris5} 0.1293 & \cellcolor{gris2} 0.0557 & \cellcolor{gris1} 0.0128 & \cellcolor{gris0} 0.0000 & \cellcolor{gris1} 0.0267 & \cellcolor{gris3} 0.0686 & \cellcolor{gris6} 0.1577  \\

4  & \cellcolor{gris6} 0.1782  & \cellcolor{gris3} 0.0778  & \cellcolor{gris1} 0.0277 & \cellcolor{gris0} 0.0000 &  \cellcolor{gris1} 0.0171  & \cellcolor{gris2} 0.0488 & \cellcolor{gris5} 0.1231 \\

5  & \cellcolor{gris6} 0.1525  & \cellcolor{gris3} 0.0758  & \cellcolor{gris1} 0.0143 & \cellcolor{gris0}  0.0000  & \cellcolor{gris1} 0.0189  &  \cellcolor{gris3} 0.0804 & \cellcolor{gris6}  0.1753 \\

6  & \cellcolor{gris6} 0.1696 & \cellcolor{gris3} 0.0871  &  \cellcolor{gris1} 0.0192 & \cellcolor{gris0}  0.0000 & \cellcolor{gris1} 0.0196 & \cellcolor{gris3}  0.0763 & \cellcolor{gris6} 0.1525  \\

7  & \cellcolor{gris6} 0.1779 & \cellcolor{gris3} 0.0877 & \cellcolor{gris1}  0.0147 & \cellcolor{gris0}  0.0000 & \cellcolor{gris1} 0.0169 & \cellcolor{gris3} 0.0809 & \cellcolor{gris6} 0.1562 \\

8 & \cellcolor{gris6} 0.1790 & \cellcolor{gris3} 0.0835 & \cellcolor{gris1} 0.0122 &  \cellcolor{gris0} 0.0000 &  \cellcolor{gris1}  0.0232 & \cellcolor{gris3} 0.0805 & \cellcolor{gris6} 0.1879 \\

9  & \cellcolor{gris6} 0.1816 & \cellcolor{gris3} 0.0828 & \cellcolor{gris1}  0.0195 & \cellcolor{gris0}  0.0000 & \cellcolor{gris1} 0.0211 & \cellcolor{gris3} 0.0833 & \cellcolor{gris6} 0.1849 \\

10  & \cellcolor{gris6} 0.1841  & \cellcolor{gris4} 0.0963 & \cellcolor{gris1}  0.0244 & \cellcolor{gris0}  0.0000 & \cellcolor{gris1} 0.0204  &   \cellcolor{gris3} 0.0842  & \cellcolor{gris6} 0.1817 \\
\hline
\end{tabular}
\caption{Relative variation for the parameter $\alpha$ of the gamma process.} \label{computation5}
\end{center}
\end{table}

\vspace{1cm}

\begin{table}[tbph]
\begin{center}
\begin{tabular}{|c|c|c|c|c|c|c|c|}
\hline
\backslashbox{$m$}{$\beta$} & 0.35 & 0.40 & 0.45 & 0.5 & 0.55 & 0.60 & 0.65 \\
\hline
2  & \cellcolor{gris5} 0.1425  &  \cellcolor{gris3} 0.0782  &  \cellcolor{gris1} 0.0232  & \cellcolor{gris0} 0.0000 &\cellcolor{gris1} 0.0122   & \cellcolor{gris3} 0.0719 & \cellcolor{gris5} 0.1486  \\
3  & \cellcolor{gris6} 0.1519  & \cellcolor{gris3} 0.0734 & \cellcolor{gris1} 0.0171 & \cellcolor{gris0} 0.0000 & \cellcolor{gris1}  0.0295 & \cellcolor{gris4} 0.0993  & \cellcolor{gris6} 0.1900 \\
4  & \cellcolor{gris6} 0.2004 & \cellcolor{gris4} 0.1072 & \cellcolor{gris1} 0.0246  & \cellcolor{gris0}  0.0000 & \cellcolor{gris1}  0.0259 & \cellcolor{gris3} 0.0657 & \cellcolor{gris5} 0.1466 \\
5  & \cellcolor{gris6} 0.1757 & \cellcolor{gris3} 0.0849  & \cellcolor{gris1} 0.0178 & \cellcolor{gris0} 0.0000 & \cellcolor{gris1}   0.0219 & \cellcolor{gris3} 0.873  & \cellcolor{gris6} 0.1956 \\
6  & \cellcolor{gris6} 0.1920  & \cellcolor{gris3} 0.0884  & \cellcolor{gris1} 0.0218 & \cellcolor{gris0}   0.0000 & \cellcolor{gris1}  0.0179 & \cellcolor{gris3} 0.0858 & \cellcolor{gris6}  0.1754 \\
7  & \cellcolor{gris6} 0.2018 & \cellcolor{gris4} 0.0975 & \cellcolor{gris1} 0.0163  & \cellcolor{gris0}   0.0000 & \cellcolor{gris1}   0.0247 & \cellcolor{gris4} 0.0956 & \cellcolor{gris6} 0.1882 \\
8  & \cellcolor{gris6} 0.1966 & \cellcolor{gris4} 0.0975 & \cellcolor{gris1} 0.0226 & \cellcolor{gris0}   0.0000 & \cellcolor{gris1}  0.0261 & \cellcolor{gris4} 0.1046 & \cellcolor{gris6} 0.2051 \\
9 & \cellcolor{gris6}  0.2028 & \cellcolor{gris4} 0.1034 & \cellcolor{gris1} 0.0213 & \cellcolor{gris0}   0.0000 & \cellcolor{gris1}  0.0245 & \cellcolor{gris4} 0.1103  & \cellcolor{gris6} 0.2114  \\
10  & \cellcolor{gris6} 0.2152  & \cellcolor{gris4} 0.1081 & \cellcolor{gris1} 0.0285 & \cellcolor{gris0}  0.0000 & \cellcolor{gris1}  0.0137 & \cellcolor{gris4}  0.1119 & \cellcolor{gris6}  0.2236 \\
\hline
\end{tabular}
\caption{Relative variation for the parameter $\beta$ of the gamma process.} \label{computation4}
\end{center}
\end{table}

\begin{table}[tbph]
\begin{center}
\begin{tabular}{|c|c|c|c|c|c|c|c|}
\hline
\backslashbox{$m$}{$\lambda$} & 0.010 & 0.015 & 0.020 & 0.025 & 0.030 & 0.035 & 0.040  \\
\hline
2  & \cellcolor{gris2} 0.0273 & \cellcolor{gris1}  0.0184 & \cellcolor{gris1} 0.0157  & \cellcolor{gris0} 0.0000 &\cellcolor{gris1} 0.0114 & \cellcolor{gris1}  0.0145 & \cellcolor{gris2} 0.0264 \\
3  & \cellcolor{gris2} 0.0222  & \cellcolor{gris1} 0.0189 & \cellcolor{gris1} 0.0128  & \cellcolor{gris0} 0.0000 & \cellcolor{gris1} 0.0082   & \cellcolor{gris1}  0.0120  & \cellcolor{gris2}  0.0203 \\
4  & \cellcolor{gris1} 0.0187 & \cellcolor{gris1} 0.0164 &  \cellcolor{gris1} 0.0130 & \cellcolor{gris0} 0.0000 &  \cellcolor{gris1} 0.0131  & \cellcolor{gris1}  0.0175  & \cellcolor{gris2}  0.0216 \\
5  & \cellcolor{gris1} 0.0153 & \cellcolor{gris1}  0.0127 & \cellcolor{gris1} 0.0091 & \cellcolor{gris0} 0.0000  & \cellcolor{gris1} 0.0055  &  \cellcolor{gris1} 0.0129  & \cellcolor{gris1} 0.0185 \\
6  & \cellcolor{gris1}  0.0141 & \cellcolor{gris1} 0.0091 &  \cellcolor{gris1} 0.0053  & \cellcolor{gris0} 0.0000  & \cellcolor{gris1}  0.0045 & \cellcolor{gris1}  0.0093  & \cellcolor{gris1} 0.0101 \\
7  & \cellcolor{gris1} 0.0134 & \cellcolor{gris1} 0.0102 & \cellcolor{gris1} 0.0054  & \cellcolor{gris0} 0.0000 & \cellcolor{gris1} 0.0067  & \cellcolor{gris1} 0.0121  & \cellcolor{gris1} 0.0138 \\
8 & \cellcolor{gris1}  0.0181 & \cellcolor{gris1} 0.0122 & \cellcolor{gris1} 0.0092 & \cellcolor{gris0} 0.0000  &  \cellcolor{gris1} 0.0099  & \cellcolor{gris1}  0.0129 & \cellcolor{gris1} 0.0173 \\
9  & \cellcolor{gris1} 0.0158 & \cellcolor{gris1} 0.0093 & \cellcolor{gris1} 0.0068 & \cellcolor{gris0} 0.0000  & \cellcolor{gris1} 0.0047 &  \cellcolor{gris1} 0.0069 & \cellcolor{gris1}  0.0114 \\
10  & \cellcolor{gris2} 0.0238 & \cellcolor{gris1} 0.0112 & \cellcolor{gris1} 0.0089 &  \cellcolor{gris0} 0.0000  & \cellcolor{gris1} 0.0058 & \cellcolor{gris1} 0.0136 & \cellcolor{gris2} 0.0240 \\
\hline
\end{tabular}
\caption{Relative variation for the parameter $\lambda$ of the exponential distribution.} \label{computation6}
\end{center}
\end{table}

\newpage

\section{Conclusions, limitations and further works}\label{conclusions}
In this paper, a multi-component  system consisting of degrading and {\color{black} non-degrading} components is analyzed. Degrading components are subject to degradation. {\color{black} Non-degrading components} fail without warning. A  maintenance policy {\color{black} at system-level} combining condition-based maintenance, opportunistic maintenance, an inspection policy and lead time to start the maintenance actions is developed. {\color{black} An analytic cost model for the maintenance strategy is developed using the semi-regenerative theory including the cost of the different maintenance actions and the reward. The reward of the system depends on the deterioration of the degrading components. Despite the semi-regenerative techniques simplify the analytic development, the expression of the kernel is quite complex}. Numerical examples are given combining Monte-Carlo simulation and meta-heuristic algorithms to find the optimal maintenance policy under a safety constraint. Sensibility analysis is given to show the robustness of the solution when different parameters vary.

{\color{black} An important limitation of this work refers to the analytic treatment of the non-degrading part. In this paper, the non-degrading part is treated as a whole with failure times modeled using constant failure rates. In the reliability literature, there are examples of the modeling of the non-degrading part using exponential distributions. For example, in \cite{Heier2}, a cooling fan can fail due to an electronic failure caused in the electronic control unit or a mechanical failure in the bearing, rotor or blades. The electronic part is modeled using exponential distributions. The use of exponential distributions to model the time between failures in the non-degrading part simplifies the analytic treatment. However, a more realistic treatment of the system would be to analyze the behavior of the non-degrading components separately. Furthermore, we assume that if the non-degrading part fails, it can consider as a new one after the maintenance. A more realistic situation would be to perform different maintenance actions in the non-degrading components depending on if they are failed or not. }

{\color{black} Another limitation of this paper refers to the structure of the system. We have also assumed in this paper that the degrading components works independently from an structural point of view. In practice, components of a system are not independent and dependence relationships can be established between them. {\color{black} It would be very interesting to carry out a study of the dependence between the different components in future works.} Different structures of multi-component systems, such as parallel systems, series systems or $k$-out-of $n$ systems can also be analyzed as further works. }

Finally, with respect to the optimization procedure, this work is focused on the analytic development. Further work should be focused on the optimization algorithms.

%


\section*{Appendix} \label{apendice}
In this Appendix, the probabilities given in Sections \ref{4.2}, \ref{4.3}, \ref{4.4}, \ref{4.5} and \ref{4.6} are obtained. 
\subsection*{Probability of a corrective maintenance in a semi-regenerative cycle (degrading components)}
Given $({\bf x},v)$, if the degrading component $i$ fails in $(T_{k-1},T_k)$, with $i \in I^m$, the following scenarios are envisioned. 
\begin{enumerate}
    \item The degrading component $i$ fails before the rest of degrading and not degrading components, that is
    $$\min(\sigma_{\bf{L-x}},Y)=\sigma_{L_i-x_i}, \hspace{0.3cm}
    T_{k-1}<\min({\bf \sigma_{M-x}}) <\sigma_{L_i-x_i}<T_k. $$
    \item The degrading component $i$ fails in the delay time, that is
    $$
    \left\{ T_{k-1}<\min(\sigma_{\bf{M-x}}), \right.  \hspace{0.3cm}
    \left. T_{k-1}<\min(\sigma_{\bf{L-x}},Y)<\sigma_{L_i-x_i}<\min(T_k,\sigma_{\bf{L-x}}+\tau,Y+\tau) \right\}.
    $$
\end{enumerate}
Let $P_{i}^{(c,1)}(T_k)$ and $P_{i}^{(c,2)}(T_k)$ be the probability of the first and second scenario respectively, we get that
$$
 P_{i}^{(c,1)}(T_k)=  \int_{T_{k-1}}^{T_k}f_{\sigma_{M_i-x_i}}(u_i)~  du_i \int_{u_i}^{T_k} f_{\sigma_{L_i-x_i}-\sigma_{M_i-x_i}}(w-u_i) \bar{F}_Y(w)  \prod_{j \neq i, j=1}^{m} G_j(T_{k-1},\infty, w)~  dw  , 
$$
where $G_j(w_1,w_2,w_3)$ is given by (\ref{G}).
For the second scenario,
\begin{eqnarray*}
&&  P_{i}^{(c,2)}(T_k) = \int_{T_{k-1}}^{T_k}\frac{-d}{dw}\left(\prod_{j=1, j \neq i}^{m} G_j(T_{k-1},\infty,w) \bar{F}_Y(w)\right) ~ dw \left( \right. \\ && \left. \int_{T_{k-1}}^{w^*}f_{\sigma_{M_i-x_i}}(u_i)  ~ du_i \int_{w}^{w^*}f_{\sigma_{L_i-x_i}-\sigma_{M_i-x_i}}(z-u_i) ~ dz, \right)
\end{eqnarray*}
where $w^*=\min(w+\tau,T_k)$ and $G_z(w_1,w_2,w_3)$ is given by Eq. (\ref{G}).
Hence, 
$$P_i^c(T_k)=P_i^{(c,1)}(T_k)+P_i^{(c,2)}(T_k).$$
\subsection*{Probability of a corrective maintenance in a semi-regenerative cycle (non-degrading {\color{black} part})}
Given $({\bf x},v)$, if the non-degrading {\color{black} part fails} in $(T_{k-1},T_k)$, with $k=1, 2, \ldots$, as in the previous subsection, two scenarios are envisioned.
\begin{enumerate}
\item The non-degrading {\color{black} part fails} before the degrading components, that is:
    $$\left\{T_{k-1}<\min(\sigma_{\bf{M-x}}), \quad T_{k-1}<Y<\min(\sigma_{\bf{L-x}},T_k) \right\}$$
    \item The non-degrading {\color{black} part fails} in the delay time to {\color{black} maintenance}, that is:
    $$\left\{T_{k-1}<\min(\sigma_{\bf{M-x}})<\min(\sigma_{\bf{L-x}})<Y<\min(\sigma_{\bf{L-x}}+\tau,T_k) \right\}$$
\end{enumerate}
Let $P^{(f,1)}(T_k)$ and $P^{(f,2)}(T_k)$ be the probability of the first and second scenario, respectively. Thus,
\begin{equation*} \label{eq1}
P^{(f,1)}(T_k)=\int_{T_{k-1}}^{T_k} f_Y(w)\prod_{j = 1}^{m} G_j(T_{k-1},\infty,w) ~ dw. 
\end{equation*}
and the probability for the second scenario is given by
\begin{eqnarray*}
 P^{(f,2)}(T_k) = \int_{T_{k-1}}^{T_k} \frac{-d}{dw} \left(\prod_{i=1}^m G_i(T_{k-1},\infty,w)\right) \left(\int_{w}^{w^*} f_Y(z)dz\right) ~ dw,    
\end{eqnarray*}
where $w^*=\min(w+\tau,T_k)$. Hence
\begin{eqnarray*}
 P^{f}(T_k) = P^{(f,1)}(T_k) +P^{(f,2)}(T_k).  
\end{eqnarray*}
\subsection*{Expected downtime in a semi-regenerative cycle}
To compute the expected downtime in a semi-regenerative cycle, first the expectation that corresponds to the expected downtime of the non-degrading {\color{black} part}
$$\mathbb{E}\left[\left(\min(Y+\tau, {\bf \sigma_{L-x}}+\tau, T_k)-Y\right)\mathbf{1}_{\left\{T_{k-1}<\min(\sigma_{\bf{M-x}}), \, \, T_{k-1}<\min({\bf \sigma_{L-x}},Y)  \leq Y<\min(T_k,Y+\tau, {\bf \sigma_{L-x}}+\tau\right\}}\right],  $$
is obtained. This expectation is evaluated considering that the non-degrading {\color{black} part} fails before the degrading components and considering that the non-degrading {\color{black} part} fails after a degrading component. Hence, the expected downtime of the non-degrading {\color{black} part} is given by 
\begin{eqnarray*}
D^{nm}(T_k) &=& \mathbb{E} \left[(\min(Y+\tau,T_k)-Y)\mathbf{1}_{\left\{T_{k-1}<\min(\sigma_{\bf{M-x}}), \, T_{k-1}<Y<\min({\bf \sigma_{L-x}})<T_k\right\}}\right] \\
&+& \mathbb{E} \left[(\min({\bf \sigma_{L-x}}+\tau,T_k)-Y)\mathbf{1}_{\left\{T_{k-1}<\min(\sigma_{\bf{M-x}}), \, T_{k-1}<\min({\bf \sigma_{L-x}})<Y<T_k\right\}}\right]
\end{eqnarray*}
and
\begin{eqnarray*}
D^{nm}(T_k) &=& \int_{T_{k-1}}^{T_k} f_Y(w)\prod_{j = 1}^{m} G_j(T_{k-1},\infty,w)(w^*-w) dw \\
&+& \int_{T_{k-1}}^{T_k} \frac{-d}{dw} \left(\prod_{i=1}^m G_i(T_{k-1},\infty,w)\right) \left(\int_{w}^{w^*} f_Y(z)dz\right) (w^*-z)~ dw, 
\end{eqnarray*}
where $w^*=\min(w+\tau,T_k)$. 

To compute the expected downtime for the component $i$, for $i \in I^m$, the expectation
$$\mathbb{E}\left[\left(\min(Y+\tau, {\bf \sigma_{L-x}}+\tau, T_k)-\sigma_{L_i-x_i}\right)\mathbf{1}_{\left\{T_{k-1}<{\bf \sigma_{M-x}}, \, \, T_{k-1}<\min(Y, {\bf \sigma_{L-x}})<\sigma_{L_i-x_i}<T_k \right\}}\right],    $$
has to be evaluated. As in the non-degrading {\color{black} part} case, the expectation is evaluated considering that the degrading component $i$ is the first component to fail and considering that other components fail before the degrading component $i$. That is
\begin{eqnarray*}
D_{i}(T_k) &=& \mathbb{E} \left[(\min(\sigma_{L_i-x_i}+\tau,T_k)-\sigma_{L_i-x_i})\mathbf{1}_{\left\{T_{k-1}<\min(\sigma_{\bf{M-x}}), \, T_{k-1}<\sigma_{L_i-x_i} \leq \min(Y,  \sigma_{\bf{L-x}})<T_k\right\}}\right] \\
&+& \mathbb{E} \left[(\min({\bf \sigma_{L-x}}+\tau,Y, T_k)-\sigma_{L_i-x_i})\mathbf{1}_{\left\{T_{k-1}<\min(\sigma_{\bf{M-x}}), \, T_{k-1}<\min(Y, {\bf \sigma_{L-x}})<\sigma_{L_i-x_i}<T_k\right\}}\right]
\end{eqnarray*}
and these probabilities are given as follows
\begin{eqnarray*}
&& D_i(T_k)  =  \int_{T_{k-1}}^{T_k}f_{\sigma_{M_i-x_i}}(u_i)~  du_i \int_{u_i}^{T_k} f_{\sigma_{L_i-x_i}-\sigma_{M_i-x_i}}(w-u_i) \bar{F}_Y(w)  \prod_{j \neq i} G_j(T_{k-1},\infty, w)~  (w^*-w)dw  \\
&& + \int_{T_{k-1}}^{T_k}\frac{-d}{dw}\left(\prod_{j=1, j \neq i}^{m} G_j(T_{k-1},\infty,w) \bar{F}_Y(w)  dw \right) \int_{T_{k-1}}^{w^*}f_{\sigma_{M_i-x_i}}(u_i)  du_i \int_{w}^{w^*}f_{\sigma_{L_i-x_i}-\sigma_{M_i-x_i}}(z-u_i)(w^*-z)dz, 
\end{eqnarray*}
where $w^*=\min(w+\tau,T_k)$. 
\subsection*{Probability of a preventive maintenance of the {\color{black} degrading} component $i$ in the semi-regenerative cycle}
To evaluate this probability, two cases are envisioned.
\begin{enumerate}
    \item The preventive threshold of the degrading component $i$ is exceeded in $(T_{k-1},T_k)$ and, at least, {\color{black} a failure of a degrading component or a failure of the non-degrading part} occurs in this interval. It corresponds to the following event
    $$\left\{T_{k-1}<\min(\sigma_{\bf{M-x}}) \leq \sigma_{M_i-x_i} \leq T_k, \, \, T_{k-1}<\min(\sigma_{\bf{L-x}},Y)<T_k<\sigma_{L_i-x_i} \right\}$$
    \item The preventive threshold of the degrading component $i$ is exceeded in $(T_{k-1},T_k)$ and {\color{black}no failure arrives to the system} in $(T_{k-1},T_k)$. 
     $$\left\{T_{k-1}<\min(\sigma_{\bf{M-x}}) \leq \sigma_{M_i-x_i} \leq T_k<\min(\sigma_{\bf{L-x}},Y) \right\}$$
\end{enumerate}
Let $P_{i}^{(p,1)}(T_k)$ be the probability of the scenario 1. This probability is given by:
\begin{eqnarray*}
&& P_{i}^{(p,1)}(T_k) = \int_{T_{k-1}}^{T_k} \frac{-d}{dw} \left(\prod_{j \neq i} \bar{G}_j(T_{k-1},\infty,w) \bar{F}_Y(w)\right) G_i(T_{k-1},T_k, w^*) ~ dw, 
\end{eqnarray*}
where $w^*=\min(w+\tau,T_k)$. On the other hand, the probability of the second scenario is given by:
\begin{eqnarray*}
P_{i}^{(p,2)}(T_k) &=& \left(\int_{T_{k-1}}^{T_k}f_{\sigma_{M_i-x_i}}(u_i)\bar{F}_{\sigma_{L_i-x_i}-\sigma_{M_i-x_i}}(T_k-u_i) du_i \right)\bar{F}_Y(T_k) \prod_{j \neq i} G_j(T_{k-1},\infty,T_k) = \\
&=& G_i(T_{k-1},T_k,T_k)\bar{F}_Y(T_k) \prod_{j \neq i} G_j(T_{k-1},\infty,T_k). 
\end{eqnarray*}
Finally, the probability of a preventive maintenance of the component $i$ is given by
\begin{equation*}
   P_i^p(T_k) =P_{i}^{(p,1)}(T_k)+P_{i}^{(p,2)}(T_k). 
\end{equation*}

\subsection*{Expected length of a semi-regenerative cycle}
To evaluate this expectation, the following cases are considered. 

\begin{eqnarray*}
\mathbb{E}[O] &=&  \sum_{k=1}^{\infty} \mathbb{E}\left[T_k \mathbf{1}_{\left\{T_{k-1}<\min({\bf \sigma_{M-x}})<T_k<\min({\bf \sigma_{L-x}}, Y)\right\}}\right] \\
&+& \sum_{k=1}^{\infty}  \mathbb{E}\left[T_k \mathbf{1}_{\left\{T_{k-1}<\min(\sigma_{M-x})<\min(\sigma_{L-x},Y)<T_k<\min(\sigma_{L-x},Y)+\tau\right\}}\right] \\ &+&
\sum_{k=1}^{\infty}  \mathbb{E}\left[(\min(\sigma_{L-x},Y)+\tau) \mathbf{1}_{\left\{T_{k-1}<\min(\sigma_{M-x}), T_{k-1}< Y <\min(\sigma_{L-x},Y)+\tau<T_k\right\}}\right]. 
\end{eqnarray*}

Hence, we get that

\begin{eqnarray*}
\mathbb{E}[O] &=& \sum_{k=1}^{\infty} T_k \bar{F}_Y(T_k)\sum_{j=1}^{m} \int_{T_{k-1}}^{T_k}f_{\sigma_{M_j-x_j}}(u_j)\prod_{i \neq j}^{m} G_i(u_j,\infty,T_k) du_j \\
&+& \sum_{k=1}^{\infty} \int_{T_{k-1}}^{T_k-\tau} f_Y(u)\prod_{i=1}^{m} G_i(T_{k-1},\infty,u)(u+\tau) du \\
&+&  \sum_{k=1}^{\infty} T_k\int_{T_{k}-\tau}^{T_k} f_Y(u)\prod_{i=1}^{m} G_i(T_{k-1},\infty,u) du \\
&+& \sum_{k=1}^{\infty} \sum_{j=1}^{m} \int_{T_{k-1}}^{T_k-\tau} f_{\sigma_{M_j-x_j}}(u_j)\int_{u_j}^{T_k-\tau}f_{\sigma_{L_j-x_l}-\sigma_{M_j-x_j}}(w_j-u_j)du_j \left( \right. \\ && \left. \prod_{i \neq j}^{m} G_i(T_{k-1},\infty,w_j)\bar{F}_Y(w_j)(w_j+\tau) dw_j \right)\\
&+& \sum_{k=1}^{\infty} \sum_{j=1}^{m} \int_{T_{k-1}}^{T_k}f_{\sigma_{M_j-x_j}}(u_j)du_j\int_{\max(T_k-\tau,u_j)}^{T_k}f_{\sigma_{L_j-x-j}-\sigma_{M_j-x_j}}(w_j-u_j)\prod_{i \neq j}^{m} G_i(T_{k-1},\infty,w_j)\bar{F}_Y(w_j)T_k.
\end{eqnarray*}

\end{document}